\documentclass{CSML}

\pdfoutput=1

\usepackage{lastpage}
\def\dOi{13(4:23)2017}
\lmcsheading%
{\dOi}
{1--\pageref{LastPage}}
{}
{}
{Mar.~15, 2016}
{Oct.~29, 2017}
{}

\usepackage[english]{babel}
\usepackage{amssymb,amsmath}
\usepackage{bbold}
\usepackage{amsrefs}
\usepackage{stackrel}
\usepackage{mathrsfs}
\usepackage[hidelinks]{hyperref}
\usepackage{cleveref}

\newcommand{\HACint}{\ensuremath{\usftext{HAC}_\intern}}

\DeclareMathOperator{\st}{st} %standard
\newcommand{\intern}{\ensuremath{{\usftext{int}}}}

\newcommand{\HAC}{\ensuremath{{\usftext{HAC}}}}
\newcommand{\R}{\ensuremath{\usftext{R}}}

\newcommand{\LLPO}{\ensuremath{\usftext{LLPO}^{\st}}}

\newcommand{\T}{\ensuremath{\mathcal{T}}}

\newcommand{\usftext}[1]{\textsf{\upshape #1}}

\newcommand{\QFAC}{\ensuremath{{\usftext{QF-AC}}}} % PA iaft
\newcommand{\ACA}{\ensuremath{\usftext{ACA}}} %
\newcommand{\RCA}{\ensuremath{\usftext{RCA}}} %

\newcommand{\ZFC}{\ensuremath{\usftext{ZFC}}}
\newcommand{\IST}{\ensuremath{\usftext{IST}}}

\newcommand{\WKL}{\ensuremath{\usftext{WKL}}}

\newcommand{\tup}{\underline} %tuple
\def\bdefi{\begin{defi}\rm}
\def\edefi{\end{defi}}
\def\bnota{\begin{nota}\rm}
\def\enota{\end{nota}}
\def\brem{\begin{rem}\rm}
\def\erem{\end{rem}}

\def\STP{\textup{STP}}

\def\RCA{\textup{RCA}}

\def\WKL{\textup{WKL}}

\def\HAC{\textup{HAC}}

\def\bye{\end{document}}

\def\P{{\mathcal  P}}

\def\N{{\mathbb  N}}
\def\Q{{\mathbb  Q}}
\def\R{{\mathbb  R}}

\def\LLPO{{\textup{LLPO}}}

\def\R{{\mathbb{R}}}
\def\({\textup{(}}
\def\){\textup{)}}

\def\st{\textup{st}}
\def\asa{\leftrightarrow}

\def\di{\rightarrow}

\def\ACA{\textup{ACA}}

\newbox\gnBoxA
\newdimen\gnCornerHgt
\setbox\gnBoxA=\hbox{$\ulcorner$}
\global\gnCornerHgt=\ht\gnBoxA
\newdimen\gnArgHgt

\def\bdefi{\begin{defi}\rm}
\def\edefi{\end{defi}}
\def\bnota{\begin{nota}\rm}
\def\enota{\end{nota}}
\def\brem{\begin{rem}\rm}
\def\erem{\end{rem}}

\def\FIVE{\Pi_{1}^{1}\text{-CA}_{0}}

\def\ATR{\textup{ATR}}

\def\STP{\textup{STP}}

\def\INT{\textup{INT}}

\def\RCA{\textup{RCA}}

\def\RCAo{\textup{RCA}_{0}^{\omega}}
\def\RCAO{\textup{RCA}_{0}^{\Omega}}

\def\WKL{\textup{WKL}}

\def\IVT{\textup{IVT}}
\def\UWKL{\textup{UWKL}}

\def\WWKL{\textup{WWKL}}
\def\UWWKL{\textup{UWWKL}}

\def\bye{\end{document}}

\def\P{{\mathcal  P}}

\def\N{{\mathbb  N}}
\def\Q{{\mathbb  Q}}
\def\R{{\mathbb  R}}

\def\MUC{\textup{MUC}}

\def\LLPO{{\textup{LLPO}}}

\def\R{{\mathbb{R}}}
\def\({\textup{(}}
\def\){\textup{)}}

\def\asa{\leftrightarrow}

\def\di{\rightarrow}

\def\ACA{\textup{ACA}}
\def\paai{\Pi_{1}^{0}\textup{-}\usftext{TRANS}}

\newcommand\be{\begin{equation}}
\newcommand\ee{\end{equation}}

\def\bdefi{\begin{defi}\rm}
\def\edefi{\end{defi}}
\def\bnota{\begin{nota}\rm}
\def\enota{\end{nota}}
\def\brem{\begin{rem}\rm}
\def\erem{\end{rem}}

\def\FIVE{\Pi_{1}^{1}\text{-\textsf{CA}}_{0}}

\def\ATR{\textup{\textsf{ATR}}}

\def\STP{\textup{\textsf{STP}}}

\def\H{\textup{\textsf{H}}}
\def\RCA{\textup{\textsf{RCA}}}

\def\RCAo{\textup{\textsf{RCA}}_{0}^{\omega}}
\def\RCAO{\textup{\textsf{RCA}}_{0}^{\Lambda}}

\def\ns{\textup{\textsf{ns}}}

\def\WKL{\textup{\textsf{WKL}}}

\def\IVT{\textup{IVT}}
\def\IVT{\textup{\textsf{IVT}}}
\def\UWKL{\textup{\textsf{UWKL}}}

\def\WWKL{\textup{\textsf{WWKL}}}
\def\UWWKL{\textup{\textsf{UWWKL}}}

\def\T{\mathcal{T}}

\def\bye{\end{document}}

\def\P{\textup{\textsf{P}}}

\def\N{{\mathbb  N}}
\def\Q{{\mathbb  Q}}
\def\R{{\mathbb  R}}

\def\BIN{\textup{\textsf{BIN}}}
\def\UBIN{\textup{\textsf{UBIN}}}
\def\UDQ{\textup{\textsf{UDQ}}}

\def\MUC{\textup{\textsf{MUC}}}

\def\LLPO{{\textup{\textsf{LLPO}}}}

\def\R{{\mathbb{R}}}
\def\({\textup{(}}
\def\){\textup{)}}

\def\st{\textup{st}}

\def\asa{\leftrightarrow}

\def\di{\rightarrow}

\def\ACA{\textup{\textsf{ACA}}}
\def\paai{\Pi_{1}^{0}\textup{-\textsf{TRANS}}}

\def\QFAC{\textup{\textsf{QF-AC}}}

\def\SCF{\textup{\textsf{SCF}}}
\def\LMP{\textup{\textsf{LMP}}}

\def\MU{\textup{\textsf{MU}}}

\def\HAC{\textup{\textsf{HAC}}}

\def\INT{\textup{\textsf{int}}}

\def\UIVT{\textup{\textsf{UIVT}}}

\numberwithin{equation}{section}
\numberwithin{thm}{section}

\begin{document}
\title[Grilliot's trick in Nonstandard Analysis]{Grilliot's trick in Nonstandard Analysis}
%\subtitle{and the nonstandard content of effective Reverse Mathematics}
\author[S.~Sanders]{Sam Sanders}
\thanks{This research was supported by the following funding bodies: FWO Flanders, the John Templeton Foundation, the Alexander von Humboldt Foundation, and the Japan Society for the Promotion of Science.}
\address{Munich Center for Mathematical Philosophy, LMU Munich, Germany \& Department of Mathematics, Ghent University}
\email{sasander@me.com}
\subjclass{Theory of computation: Computability \& Proof theory}

\begin{abstract}
Recently, Dag Normann and the author have established a connection between \emph{higher-order computability theory} and \emph{Nonstandard Analysis}; new results in both fields are obtained by exploiting this connection (\cite{dagsam}).
Now, on the side of computability theory, the technique known as \emph{Grilliot's trick} constitutes a template for explicitly defining the Turing jump functional $(\exists^{2})$ in terms of a given effectively discontinuous type two functional (\cite{grilling}).  
In light of the aforementioned connection, it is a natural question what corresponds to Grilliot's trick in Nonstandard Analysis?
In this paper, we discuss the \emph{nonstandard extensionality trick}: a technique similar to Grilliot's trick {in Nonstandard Analysis}.  This nonstandard trick proceeds by deriving from the existence of certain \emph{nonstandard discontinuous} functionals, the \emph{Transfer} principle from Nonstandard analysis limited to $\Pi_{1}^{0}$-formulas; from this (generally ineffective) implication, we obtain an effective implication expressing the Turing jump functional in terms of a discontinuous functional (and no longer involving Nonstandard Analysis).  The advantage of our nonstandard approach is that \emph{one obtains effective content without paying attention to effective content}.  We also discuss a new class of functionals which fall outside the established categories.  
These functionals derive from the \emph{Standard Part} axiom of Nonstandard Analysis.   
%We also discuss a new class of functional
\end{abstract}

%\setcounter{page}{0}
%\tableofcontents
%\thispagestyle{empty}
%\newpage

\maketitle

\section{Introduction}\label{intro}
Recently, Dag Normann and the author have established a connection between \emph{higher-order computability theory} and \emph{Nonstandard Analysis} (\cite{dagsam}).  
In the latter, they investigate the complexity of functionals connected to the \emph{Heine-Borel compactness} of Cantor space.   Surprisingly, this complexity turns out to be 
intimately connected to the \emph{nonstandard compactness} of Cantor space as given by \emph{Robinson's theorem} (See \cite{loeb1}*{p.\ 42}) in Nonstandard Analysis.
In fact, the results in \cite{dagsam} are `holistic' in nature in that theorems in computability theory give rise to theorems in Nonstandard Analysis, \emph{and vice versa}.  
We discuss these results in Section \ref{knowledge} as they serve as the motivation for this paper.  % as follows.     
   
\medskip

In light of the aforementioned connection, it is a natural question which notions from higher-order computability theory have elegant analogues in Nonstandard Analysis, \emph{and vice versa}.
This paper explores one particular case of this question, namely for the technique know as \emph{Grilliot's trick}, introduced in \cite{grilling}.  
The latter `trick' actually constitutes a template for explicitly defining the Turing jump functional $(\exists^{2})$ in terms of a given effectively discontinuous type two functional.  
% (\cite{grilling}).  
Below, we introduce the \emph{nonstandard extensionality trick}, which is a technique similar to {Grilliot's trick} {in Nonstandard Analysis}.  
In this way, we study a new \emph{computational} aspect of Nonstandard Analysis pertaining to {Reverse Mathematics} (RM), in line with the results in \cites{samzoo,samzooII, samGH, sambon}.
We refer to \cites{simpson2, simpson1} for an overview to RM.   
We shall make use of \emph{internal set theory}, i.e.\ Nelson's axiomatic Nonstandard Analysis (\cite{wownelly}).  % and its fragments from \cite{brie}.  
We introduce internal set theory and its fragments from \cite{brie} in Section~\ref{WIST}.  

\medskip

The nonstandard extensionality trick sums up as: From the existence of \emph{nonstandard discontinuous} functionals, the \emph{Transfer} principle from Nonstandard analysis (See Section~\ref{IST}) limited to $\Pi_{1}^{0}$-formulas is derived; from this (generally ineffective) implication, we obtain an effective implication expressing the Turing jump functional in terms of a discontinuous functional (and no longer involving Nonstandard Analysis).  Essential to obtaining this effective implication is the `term extraction theorem' in Theorem \ref{consresult}, based on \cite{brie}.  
We shall apply the nonstandard extensionality trick to \emph{binary expansion}, the \emph{intermediate value theorem}, the \emph{Weierstra\ss~maximum theorem}, and \emph{weak weak K\"onig's lemma} in Section \ref{main}. 

\medskip

Now, combining the aforementioned results with similar results from \cites{samzoo, samzooII} regarding the RM zoo from \cite{damirzoo}, one gets the idea that the 
higher-order landscape is not very rich and similar to the second-order framework.  
To counter this view, we discuss a new class of functionals in Section~\ref{knowledge} which do not fit the existing categories of RM.  
These functionals are inspired by the \emph{Standard Part} axiom of Nonstandard Analysis.  % (See Section \ref{IST}).    

\medskip

Finally, we hasten to point out that there are well-established techniques for obtaining effective content from classical mathematics, the most prominent one being the \emph{proof mining} program (\cite{kohlenbach3}).  
In particular, the effective results in this paper could be or have been obtained in this way.  
What is surprising about the results in this paper (in our opinion) is the emergence of effective content (with relative ease) from Nonstandard Analysis \emph{despite} claims 
that the latter is somehow fundamentally non-constructive by e.g.\ Bishop and Connes (See \cite{samsynt} for a detailed discussion of the Bishop-Connes critique).

\section{Internal set theory and fragments}\label{WIST}
In this section, we sketch \emph{internal set theory}, Nelson's \emph{syntactic} approach to Nonstandard Analysis, first introduced in \cite{wownelly}, and its fragments from \cite{brie}.  An in-depth and completely elementary introduction to the constructive content of Nonstandard Analysis is \cite{SB}.      
\subsection{Introducing internal set theory}\label{IST}
Nelson's system $\IST$ of internal set theory is defined as follows:  The language of $\IST$ consist of the language of $\ZFC$, the `usual' foundations of mathematics, plus a new predicate `st($x$)', read as `$x$ is standard'.  %The following conventions are used.  
% is added to the language of \textsf{ZFC}, the usual foundations of mathematics.  
The new quantifiers $(\forall^{\st}x)(\dots)$ and $(\exists^{\st}y)(\dots)$ are short for $(\forall x)(\st(x)\di \dots)$ and $(\exists y)(\st(y)\wedge \dots)$.  
A formula of $\IST$ is called \emph{internal} if it does not involve `st', and \emph{external} otherwise.   

\medskip

The system $\IST$ is the internal system $\ZFC$ plus the
following\footnote{The `fin' in \textsf{(I)} means that $x$ is finite,
  i.e.\ its number of elements are bounded by a natural number.} three
external axioms \emph{Idealisation}, \emph{Standard Part}, and
\emph{Transfer} which govern the predicate `st'.

%They are respectively defined\footnote{The superscript `fin' in \textsf{(I)} means that $x$ is finite, i.e.\ its number of elements are bounded by a natural number.} as:  
\begin{itemize}[label=\textsf{(S)}]
\item[\textsf{(I)}] $(\forall^{\st~\textup{fin}}x)(\exists y)(\forall z\in x)\varphi(z,y)\di (\exists y)(\forall^{\st}x)\varphi(x,y)$, for internal $\varphi$.  % with any (possibly nonstandard) parameters.  
\item[\textsf{(S)}] $(\forall^{\st} x)(\exists^{\st}y)(\forall^{\st}z)\big((z\in x\wedge \varphi(z))\asa z\in y\big)$, for any formula $\varphi$.
\item[\textsf{(T)}] $(\forall^{\st}t)\big[(\forall^{\st}x)\varphi(x, t)\di (\forall x)\varphi(x, t)\big]$, for internal $\varphi$ and $t, x$ the only free variables.  
\end{itemize}
%Nelson's system \textsf{IST} is the internal system \textsf{ZFC} extended with the aforementioned external axioms;  
Nelson proves in \cite{wownelly} that \textsf{IST} is a conservative extension of \textsf{ZFC}, i.e.\ $\ZFC$ and $\IST$ prove the same (internal) sentences.    
Various fragments of $\IST$ have been studied previously, and we shall make essential use of the system $\P$, a fragment of $\IST$ based on Peano arithmetic, introduced in Section~\ref{PIPI}.  
%Before that, we briefly introduce system $\textsf{T}$ in Section \ref{TITI}.  
The system $\P$ was first introduced in \cite{brie} and is exceptional in that it has a `term extraction procedure' \emph{with a very wide scope}.  We discuss this aspect of $\P$ in more detail in Remark~\ref{firliborn}.     

\subsection{The classical system $\P$}\label{PIPI}
In this section, we introduce the classical system $\P$ which is a conservative extension of Peano arithmetic by Theorem \ref{consresult}.  
We refer to \cite{kohlenbach3}*{\S3.3} for the detailed definition of the rather mainstream system $\textsf{E-PA}^{\omega}$, i.e.\ \emph{Peano arithmetic in all finite types with the axiom of extensionality}. 
The system $\P$ consist of the following axioms, starting with the basic ones.  
\bdefi\label{debs}[Basic axioms of $\P$]
\begin{enumerate}
\item The system \textsf{E-PA}$^{\omega*}$ is the definitional extension of \textsf{E-PA}$^{\omega}$ with types for finite sequences as in \cite{brie}*{\S2}. 
\item The set $\T^{*}$ is the collection of all the constants in the language of $\textsf{E-PA}^{\omega*}$.  
\item The external induction axiom \textsf{IA}$^{\st}$ is  % as follows.  
\be\tag{\textsf{IA}$^{\st}$}
\Phi(0)\wedge(\forall^{\st}n^{0})(\Phi(n) \di\Phi(n+1))\di(\forall^{\st}n^{0})\Phi(n).     
\ee
\item\label{krafoi} The system $ \textsf{E-PA}^{\omega*}_{\st} $ is defined as $ \textsf{E-PA}^{\omega{*}} + \T^{*}_{\st} + \textsf{IA}^{\st}$, where $\T^{*}_{\st}$
consists of the following basic axiom schemas.
\begin{enumerate}
\item The schema\footnote{The language of $\textsf{E-PA}_{\st}^{\omega*}$ contains a symbol $\st_{\sigma}$ for each finite type $\sigma$, but the subscript is always omitted.  Hence $\T^{*}_{\st}$ is an \emph{axiom schema} and not an axiom.\label{omit}} $\st(x)\wedge x=y\di\st(y)$. \label{komit}
\item The schema providing for each closed term $t\in \T^{*}$ the axiom $\st(t)$.
\item The schema $\st(f)\wedge \st(x)\di \st(f(x))$.
\end{enumerate}
\end{enumerate}
Secondly, Nelson's axiom \emph{Standard part} is weakened in \cite{brie} to $\HAC_{\INT}$:
\be\tag{$\HAC_{\INT}$}
(\forall^{\st}x^{\rho})(\exists^{\st}y^{\tau})\varphi(x, y)\di (\exists^{\st}F^{\rho\di \tau^{*}})(\forall^{\st}x^{\rho})(\exists y^{\tau}\in F(x))\varphi(x,y),
\ee
where $\varphi$ is any internal formula and $\tau^{*}$ is the type of finite sequences of objects of type $\tau$.  Note that $F$ only provides a \emph{finite sequence} of witnesses to $(\exists^{\st}y)$, explaining the name \emph{Herbrandized Axiom of Choice} for $\HAC_{\INT}$.

\medskip
      
Thirdly,  Nelson's axiom idealisation \textsf{I} appears in \cite{brie} as follows:  
\be\tag{\textsf{I}}
(\forall^{\st} x^{\sigma^{*}})(\exists y^{\tau} )(\forall z^{\sigma}\in x)\varphi(z,y)\di (\exists y^{\tau})(\forall^{\st} x^{\sigma})\varphi(x,y), 
\ee
where $\varphi$ is internal and $\sigma^{*}$ is the type of finite sequences of objects of type $\sigma$.
\edefi
%Hence, the system $ \textsf{E-PA}^{\omega*}_{\st} $ is a trivial extension of \textsf{E-PA}$^{\omega^{*}}$.  
For $\P\equiv \textsf{E-PA}^{\omega*}_{\st} +\HAC_{\INT} +\textsf{I}$, we have the following `term extraction theorem', which is not explicitly formulated or proved in \cite{brie}. A proof may be found in \cites{samzoo, sambon}.     
\begin{thm}[Term extraction]\label{consresult}
Let $\varphi$ be an internal formula and let $\Delta_{\intern}$ be a collection of internal formulas.  If we have:
\be\label{antecedn}
\P + \Delta_{\intern} \vdash (\forall^{\st}\underline{x})(\exists^{\st}\underline{y})\varphi(\underline{x}, \underline{y}, \underline{a})
\ee
then one can extract from the proof a sequence of closed terms $t$ in $\mathcal{T}^{*}$ such that
\be\label{consequalty}
\textup{\textsf{E-PA}}^{\omega*} + \Delta_{\intern} \vdash ( \forall \tup x) (\exists \tup y\in \tup t(\tup x)) \varphi(\tup x,\tup y, \tup a).
\ee
\end{thm}
\begin{proof}
The proof of the theorem in a nutshell: A proof interpretation $S_{\st}$ is defined in \cite{brie}*{Def.\ 7.1}; a tedious but straightforward verification using the clauses (i)-(v) in \cite{brie}*{Def.\ 7.1} establishes that $\Phi(\underline{a})^{S_{\st}}\equiv \Phi(\underline{a})$ for $\Phi(\underline{a})\equiv (\forall^{\st}\underline{x})(\exists^{\st}\underline{y})\varphi(\underline{x}, \underline{y}, \underline{a})$ and $\varphi$ internal.  The theorem now follows immediately from \cite{brie}*{Theorem 7.7}.  
\end{proof}
The term $t$ in \eqref{consequalty} is \emph{primitive recursive} in the sense of G\"odel's system ${T}$.  The latter was introduced in \cite{godel3}, and is also discussed in \cite{kohlenbach3}*{\S3}.  For the rest of this paper, a `normal form' will refer to a formula as in \eqref{antecedn}, i.e.\ of the form $ (\forall^{\st} \tup x )(\exists^{\st} \tup y) \varphi(\tup x,\tup y, \tup a)$ for internal $\varphi$.

\medskip

As expected, the previous theorem does not really depend on the presence of full Peano arithmetic.  
Indeed, let \textsf{E-PRA}$^{\omega}$ be the system defined in \cite{kohlenbach2}*{\S2} and let \textsf{E-PRA}$^{\omega*}$ 
be its definitional extension with types for finite sequences as in \cite{brie}*{\S2}.  We permit ourselves a slight abuse of notation by not distinguishing between Kohlenbach's $\RCAo\equiv \textup{\textsf{E-PRA}}^{\omega}+\QFAC^{1,0}$ (See \cite{kohlenbach2}*{\S2}) and $\textup{\textsf{E-PRA}}^{\omega*}+\QFAC^{1,0}$.  
\begin{cor}\label{consresultcor2}
The previous theorem and corollary go through for $\P$ and $\textup{\textsf{E-PA}}^{\omega*}$ replaced by $\RCAO\equiv \textsf{\textup{E-PRA}}^{\omega*}+\T_{\st}^{*} +\HAC_{\INT} +\textsf{\textup{I}}+\QFAC^{1,0}$ and $\RCAo$.  
\end{cor}
\begin{proof}
The proof of \cite{brie}*{Theorem 7.7} goes through for any fragment of \textsf{E-PA}$^{\omega{*}}$ which includes \textsf{EFA}, sometimes also called $\textsf{I}\Delta_{0}+\textsf{EXP}$.  
In particular, the exponential function is (all what is) required to `easily' manipulate finite sequences.    
\end{proof}
Next, we discuss the vast scope of the term extraction result in Theorem~\ref{consresult}.  
\begin{rem}[The scope of term extraction]\label{firliborn}\rm
First of all, there are examples of classically provable sentences (See \cite{kohlenbach3}*{\S2.2}) with \emph{only two} quantifier alternations from which no computational information can be extracted.  
By contrast, it is shown in \cite{sambon} that the scope of Theorem~\ref{consresult} encompasses all theorems of `pure' Nonstandard Analysis, where `pure' means that only \emph{nonstandard} definitions (of continuity, compactness, differentiability, Riemann integration, et cetera) are used.  Indeed, it is easy to prove in $\P$ (or $\RCAO$) that these nonstandard definitions have equivalent normal forms, and that an implication between two normal forms is again equivalent to a normal form.  In other words, the scope of the term extraction result in Theorem~\ref{consresult} is vast, as explored in \cites{sambon, samzoo, samzooII, samGH}.  
\end{rem}
Finally, we note that a `constructive' version of $\P$ is introduced in \cite{brie}*{\S5}.  
In particular, the system $\H$ is a conservative extension of Heyting arithmetic $\textsf{E-HA}^{\omega}$ for the latter's language, and satisfies 
a term extraction theorem similar to Theorem~\ref{consresult} (See \cite{brie}*{Theorem~5.9}).  We briefly discuss $\H$ in Remark \ref{reasonsofspace}.  % for reasons of space.    

\subsection{Notations}
We mostly use the notations from \cite{brie}, some of which we repeat.  

\medskip

First of all, the following notations were sketched in Section \ref{IST}.  
\begin{rem}[Notations]\label{notawin}\rm
We write $(\forall^{\st}x^{\tau})\Phi(x^{\tau})$ and $(\exists^{\st}x^{\sigma})\Psi(x^{\sigma})$ as short for the formula
$(\forall x^{\tau})\big[\st(x^{\tau})\di \Phi(x^{\tau})\big]$ and $(\exists x^{\sigma})\big[\st(x^{\sigma})\wedge \Psi(x^{\sigma})\big]$.     
We also write $(\forall x^{0}\in \Omega)\Phi(x^{0})$ and $(\exists x^{0}\in \Omega)\Psi(x^{0})$ as short for 
$(\forall x^{0})\big[\neg\st(x^{0})\di \Phi(x^{0})\big]$ and $(\exists x^{0})\big[\neg\st(x^{0})\wedge \Psi(x^{0})\big]$.  
Finally, a formula $A$ is `internal' if it does not involve $\st$.  The formula $A^{\st}$ is defined from $A$ by appending `st' to all quantifiers (except bounded number quantifiers).    
\end{rem}
Secondly, we use the usual notations for rational and real numbers in $\RCAo$ as introduced in \cite{kohlenbach2}*{p.\ 288-289}.  
We repeat some of the latter definitions.  % (and \cite{simpson2}*{I.8.1} for the former).  
\begin{defi}[Real numbers and related notions]\label{keepinitreal}\rm~
\begin{enumerate}
\item A (standard) rational $q^{0}$ is a fraction $\pm\frac{m}{n}$ for (standard) $n^{0}>0$ and (standard) $m^{0}$.  We write `$q^{0}\in \Q$' to denote that $q$ is a rational.           
\item A (standard) real number $x$ is a (standard) fast-converging Cauchy sequence $q_{(\cdot)}^{1}$, i.e.\ $(\forall n^{0}, i^{0})(|q_{n}-q_{n+i})|<_{0} \frac{1}{2^{n}})$.  
We use Kohlenbach's `hat function' from \cite{kohlenbach2}*{p.\ 289} to guarantee that every sequence $f^{1}$ is a real.  
\item We write $[x](k):=q_{k}$ for the $k$-th approximation of a real $x^{1}=(q^{1}_{(\cdot)})$.    
\item Two reals $x, y$ represented by $q_{(\cdot)}$ and $r_{(\cdot)}$ are \emph{equal}, denoted $x=_{\R}y$, if $(\forall n>0)(|q_{n}-r_{n}|\leq \frac{1}{2^{n-1}})$. Inequality $x<_{\R}y$ is defined by $(\exists n>0)(q_{n}+ \frac{1}{2^{n-1}}< r_{n})$.         
\item We  write $x\approx y$ if $(\forall^{\st} n>0)(|q_{n}-r_{n}|\leq \frac{1}{2^{n-1}})$ and $x\gg y$ if $x>_{\R}y\wedge x\not\approx y$.  
\item Functions $F:\R\di \R$ mapping reals to reals are represented by functionals $\Phi^{1\di 1}$ mapping equal reals to equal reals, i.e. 
\be\tag{\textsf{RE}}\label{furg}
(\forall x^{1}, y^{1})(x=_{\R}y\di \Phi(x)=_{\R}\Phi(y)).
\ee   
\item Sets of objects of type $\rho$ are denoted $X^{\rho\di 0}, Y^{\rho\di 0}, Z^{\rho\di 0}, \dots$ and are given by their characteristic functions $f^{\rho\di 0}_{X}$, i.e.\ $(\forall x^{\rho})[x\in X\asa f_{X}(x)=_{0}1]$, where $f_{X}^{\rho\di 0}$ is assumed to output zero or one.  
\end{enumerate}
\end{defi}

\noindent Thirdly, we use the usual extensional notion of equality in $\P$.    
\begin{rem}[Equality in $\P$]\label{equ}\rm
Equality between natural numbers `$=_{0}$' is a primitive.  Equality `$=_{\tau}$' for type $\tau$-objects $x,y$ is then defined as follows:
\be\label{aparth}
[x=_{\tau}y] \equiv (\forall z_{1}^{\tau_{1}}\dots z_{k}^{\tau_{k}})[xz_{1}\dots z_{k}=_{0}yz_{1}\dots z_{k}]
\ee
if the type $\tau$ is composed as $\tau\equiv(\tau_{1}\di \dots\di \tau_{k}\di 0)$.
In the spirit of Nonstandard Analysis, we define `approximate equality $\approx_{\tau}$' as follows:
\be\label{aparth2}
[x\approx_{\tau}y] \equiv (\forall^{\st} z_{1}^{\tau_{1}}\dots z_{k}^{\tau_{k}})[xz_{1}\dots z_{k}=_{0}yz_{1}\dots z_{k}]
\ee
with the type $\tau$ as above.  
The system $\P$ includes the \emph{axiom of extensionality}: %for all $\varphi^{\rho\di \tau}$ as follows:
\be\label{EXT}\tag{\textsf{E}}  
(\forall  x^{\rho},y^{\rho}, \varphi^{\rho\di \tau}) \big[x=_{\rho} y \di \varphi(x)=_{\tau}\varphi(y)   \big].
\ee
However, as noted in \cite{brie}*{p.\ 1973}, the so-called axiom of \emph{standard} extensionality \eqref{EXT}$^{\st}$ is not included in $\P$, as this would jeopardise the term extraction property as in Theorem \ref{consresult}.   
Finally, a functional $\Xi^{ 1\di 0}$ is called an \emph{extensionality functional} for $\varphi^{1\di 1}$ if 
\be\label{turki}
(\forall k^{0}, f^{1}, g^{1})\big[ \overline{f}\Xi(f,g, k)=_{0}\overline{g}\Xi(f,g,k) \di \overline{\varphi(f)}k=_{0}\overline{\varphi(g)}k \big],  
\ee
i.e.\ $\Xi$ witnesses \eqref{EXT} for $\varphi$.  
As will become clear in Section \ref{main}, $\eqref{EXT}^{\st}$ is translated to the existence of an extensionality functional when applying Theorem \ref{consresult}.      
\end{rem} 

\section{An analogue of Grilliot's trick in Nonstandard Analysis}\label{main}
In this section, we show how from certain equivalences in Nonstandard Analysis involving a fragment of Nelson's \emph{Transfer}, namely $\paai$, one obtains \emph{effective} RM-equivalences involving $(\exists^{2})$ in Kohlenbach's higher-order RM. 
\be\tag{$\exists^{2}$}
(\exists \varphi^{2})(\forall f^{1})\big[ (\exists n)(f(n)=0)\asa \varphi(f)=0   \big].
\ee
\be\tag{$\paai$}
\qquad\qquad\quad\enspace(\forall^{\st}f^{1})\big[(\forall^{\st}n)f(n)\ne 0 \di (\forall n)f(n)\ne0   \big].
\ee
To this end, we shall make use of a technique from Nonstandard Analysis we call \emph{the nonstandard extensionality trick}, and which is similar to \emph{Grilliot's trick}.  We first introduce some of the above italicised notions in the following section.
\subsection{Preliminaries}
In this section, we introduce the notion of `effective implication' and so-called Grilliot's trick.   
First of all, the notion of `effective implication' is defined as one would expect in $\RCAo$.  
\bdefi[Effective implication]\label{effimp}
An implication $(\exists \Phi)A(\Phi)\di (\exists \Psi)B(\Psi)$ (proved in $\RCA_{0}$) is \emph{effective} if there is a term $t$ (in the language of $\RCAo$) such that additionally $(\forall \Phi)[A(\Phi)\di B(t(\Phi))]$ (proved in $\RCAo$).  % i.e.\ $\Psi$ can be effectively defined in terms of $\Phi$.    
\edefi
The terms obtained using Theorem \ref{consresult} are \emph{primitive recursive} in the sense of G\"odel's system ${T}$, as discussed in Section \ref{PIPI}.  
In light of the elementary nature of an extensionality functional (See Remark~\ref{equ}), we still refer to an implication as `effective', if the term $t$ as in Definition \ref{effimp} involves an extensionality functional.  %This assumption is mostly innocent by \cite{kooltje}*{Remark 3.6}.    
Note that $\RCAo$ proves the existence of an extensionality functional thanks to $\QFAC^{1,0}$ while an unbounded search (available in a more general setting than G\"odel's $T$) also yields such a functional.  

\medskip

Secondly, as to methodology, we shall make use of a nonstandard technique, called \emph{the nonstandard extensionality trick} (See Remark \ref{trick}), similar to \emph{Grilliot's trick}.  
Now, the latter trick is in fact an explicit construction to obtain the Turing jump functional $(\exists^{2})$ from a given effectively discontinuous functional 
(See e.g.\ \cite{grilling}, \cite{kohlenbach2}*{Prop.~3.7}, or \cite{kooltje}*{Prop.\ 3.4} for more details).  In our \emph{nonstandard} trick, one obtains $\paai$ from a functional $\Phi^{1\di 1}$ which is \emph{nonstandard} discontinuous, i.e.\ there are
$x_{0}\approx_{1} x_{1}$ such that $ \Phi(x_{0})\not\approx_{1} \Phi(x_{1})$.  By applying term extraction as in Theorem~\ref{consresult}, one then obtains an effective implication involving $(\exists^{2})$.
As we will see, the nonstandard proof involving $\paai$ uses proof by contradiction, i.e.\ \emph{no attempt to obtain effective content is made in the nonstandard proofs}.  

\medskip

Thirdly, in the next sections, we apply the aforementioned nonstandard extensionality trick to \emph{binary expansion}, the \emph{intermediate value theorem}, the \emph{Weierstra\ss~maximum theorem}, and \emph{weak weak K\"onig's lemma}.  We choose these theorems due to their `non-constructive' nature, and as some of the associated uniform versions (sometimes involving sequences) have been studied (\cites{sayo, polahirst, yamayamaharehare}, \cite{simpson2}*{IV.2.12}, \cite{kohlenbach2}*{\S3}).  
It is particularly interesting that we can `recycle' the Brouwerian counterexamples to the intermediate value theorem and Weierstra\ss~maximum theorem (\cite{beeson1}*{I.7}, \cite{mandje2}) to obtain nonstandard equivalences.  

\medskip

Finally, it is a natural question \emph{why} we can obtain computational information from proofs in classical Nonstandard Analysis \emph{at all}.  
Indeed, Bishop and Connes have made rather strong claims regarding the non-constructive nature of Nonstandard Analysis (See\footnote{The third reference is Bishop's review of Keisler's introduction to Nonstandard Analysis \cite{keisler3}.} \cite{kluut}*{p.\ 513}, \cite{bishl}*{p.\ 1}, \cite{kuddd}, \cite{conman2}*{p.\ 6207} and \cite{conman}*{p.\ 26}).  Furthermore, there are examples of classically provable sentences (See \cite{kohlenbach3}*{\S2.2}) with \emph{only two} quantifier alternations from which no computational information can be extracted, and the aforementioned theorems involve a lot more quantifier alternations.  Moreover, our nonstandard proofs make use of `proof by contradiction', i.e.\ no attempt at a `constructive' proof is made.
Nonetheless, in Sections \ref{bincco} to \ref{hivt}, we shall obtain effective equivalences from certain `non-constructive' nonstandard equivalences.  
Following similar results in Section \ref{carmichael}, we offer an explanation why Nonstandard Analysis contains so much computational information.  

\subsection{Binary conversion}\label{bincco}
In this section, we study the principle of \emph{binary conversion}, i.e.\ the statement that every real can be represented in binary as follows:
\be\label{bin}\tag{$\BIN$}\textstyle
(\forall x\in [0,1])(\exists\alpha^{1}\leq_{1}1)(x=_{\R}\sum_{i=1}^{\infty}\frac{\alpha(i)}{2^{i}}).
\ee
Hirst shows in \cite{polahirst} that  $\RCA_{0}$ proves $\BIN$, and that a uniform version of the latter \emph{involving sequences} is equivalent to $\WKL$.  
Furthermore, $\BIN$ is equivalent to $\LLPO$ in \emph{constructive} Reverse Mathematics (\cite{bridges1}*{p.\ 10}), while a finer classification may be found in \cite{bergske}.  
We study a \emph{higher-type} {uniform} version of $\BIN$:
\be\label{ubin}\tag{$\UBIN$}\textstyle
(\exists \Phi:\R\di 1)(\forall x\in [0,1])\big[ \Phi(x)\leq_{1}1\wedge  x=_{\R}\sum_{i=1}^{\infty}\frac{\alpha(i)}{2^{i}}\big].
\ee
We shall first establish a particular nonstandard equivalence involving $\paai$, and a nonstandard version of $\UBIN$.  As a result of applying Corollary~\ref{consresultcor2} to this nonstandard equivalence, we obtain an \emph{effective} equivalence between $\UBIN$ and the following version of arithmetical comprehension.    
\be\tag{$\mu^{2}$}\label{Frak}
(\exists \mu^{2})\big[(\forall f^{1})\big((\exists x^{0})f(x)=0 \di  f(\mu(f))=0 \big)\big].
\ee  
The functional $(\mu^{2})$ is also known as \emph{Feferman's non-constructive mu-operator} (See \cite{avi2}*{\S8.2}), and is equivalent to $(\exists^{2})$ in $\RCAo$ by \cite{kooltje}*{\S3}.  
We denote by $\MU(\mu)$ the formula in square brackets in $(\mu^{2})$.  
We  use the following nonstandard version of $\UBIN$, called $\UBIN^{+}$:
%\be\label{ubinplus}\tag{$\UBIN^{+}$}
\[
(\exists^{\st} \Phi:\R\di 1)(\forall^{\st} x\in [0,1])\big[\UBIN(\Phi, x)\wedge (\forall^{\st}x, y\in [0,1])(x\approx y \di \Phi(x)\approx_{1}\Phi(y))  \big]. 
\]
where $\UBIN(\Phi, x)$ is the formula in square brackets in $\UBIN$.  Note that the second conjunct of $\UBIN^{+}$ expresses that $\Phi$ is `standard extensional', i.e.\ satisfies extensionality as in \eqref{furg} relative to `st' (and the range is Baire space instead of $\R$).      
\begin{thm}\label{proto7}
From a proof in $\RCAO$ that $\UBIN^{+}\asa \paai$, terms $s,t$ can be extracted such that $\RCAo$ proves:
\be\label{frood7}
(\forall \mu^{2})\big[\textsf{\MU}(\mu)\di \UBIN(s(\mu)) \big] \wedge (\forall \Phi^{1\di 1})\big[ \UBIN(\Phi)\di  \MU(t(\Phi, \Xi))  \big].
\ee
where $\Xi$ is an extensionality functional for $\Phi$ and $\UBIN(\Phi)$ is $(\forall x\in [0,1])\UBIN(\Phi, x)$. 
\end{thm}
\begin{proof}
First of all, we prove that $\UBIN^{+}\di \paai$ in $\RCAO$, and obtain the associated second conjunct of \eqref{frood7}.  The remaining results are then sketched.

\medskip

To prove $\UBIN^{+}\di \paai$ in $\RCAO$, assume $\UBIN^{+}$ and suppose that $\paai$ is false, i.e.\ there is \emph{standard} $g$ such that $(\forall^{\st}n)g(n)=0 $ but also $ (\exists m^{0})g(m)\ne 0$.  Now define the \emph{standard} sequence $\alpha_{0}$ as follows
\be\label{sly}
\alpha_{0}(i):=
\begin{cases}
0 &  (\forall n\leq i)g(n)=0 \\
1 & \text{otherwise}
\end{cases}.
\ee
Furthermore, define the \emph{standard} reals $x_{\pm}:= \frac{1}{2}\pm \sum_{n=1}^{\infty} \frac{\alpha_{0}(n)}{2^{n}}$ and note that $x_{+}\approx x_{-}$ by the definition of $g$.  Since $x_{-}<_{\R} \frac12 <_{\R}x_{+}$, the binary expansion $\alpha_{\pm}$ of $x_{\pm}$ must be such that $\alpha_{-}(1)=0$ and $\alpha_{+}(1)=1$.  However, this implies that $\Phi(x_{-})(1)=0\ne 1=\Phi(x_{+})(1)$, and also $\Phi(x_{-})\not\approx_{1}\Phi(x_{+})$.    
Clearly, the latter contradicts the standard extensionality of $\Phi$ as $x_{+}\approx x_{-}$ was also proved.    
In light of this contradiction, we must have $\UBIN^{+}\di \paai$.  

\medskip

We now prove the second conjunct in \eqref{frood7}.  %The first conjunct is proved in exactly the same way.  
Note that $\paai$ can easily be brought into the following normal form: 
\be\label{frux}
(\forall^{\st}f^{1})(\exists^{\st}i^{0})\big[(\exists n^{0})f(n)=0\di (\exists m\leq i)f(m)=0\big], 
\ee
where the formula in square brackets is abbreviated by $B(f, i)$.  Similarly, the second conjunct of $\UBIN^{+}$ has the following normal form:  
\be\label{kurve2}\textstyle
(\forall^{\st} x^{1}, y^{1}\in [0,1], k^{0})(\exists^{\st} N)\big[|x-y|<\frac{1}{N} \di \overline{\Phi(x)}k=_{0}\overline{\Phi(y)}k\big],
\ee
which is immediate by resolving `$\approx_{1}$' and `$\approx$', and bringing standard quantifiers outside. 
We denote the formula in square brackets in \eqref{kurve2} by $A(x, y, N, k, \Phi)$.  
Hence, $\UBIN^+\di \paai$ now easily yields:
\begin{align}
(\forall^{\st}\Phi, \Xi)\big[ [(\forall^{\st}x\in [0,1])\UBIN(\Phi, x)&\wedge (\forall^{\st} x,y\in [0,1], k^{0})A(x, y, \Xi(x, y, k), k, \Phi)]\notag\\
&  \di     (\forall^{\st}f^{1})(\exists^{\st}n)B(f,n)\big],  \label{tochtag}
\end{align}
as standard $\Xi$ as in the antecedent of \eqref{tochtag} yields standard outputs for standard inputs, and hence \eqref{kurve2} follows.  
Dropping the `st' in the antecedent of \eqref{tochtag} and bringing out the remaining standard quantifiers, we obtain the normal form:
\begin{align}
(\forall^{\st}\Phi, \Xi ,f^{1})(\exists^{\st}n)\big[ [(\forall x^{1}\in [0,1])&\UBIN(\Phi, x)\wedge\notag \\
& (\forall  U^{1}, S^{1}, k^{0})A(U, S, \Xi(U, S, k), k, \Phi)]  \di    B(f,n)\big]. \label{enoka}
\end{align}
Let $C(\Phi, \Xi, f, n)$ be the formula in big square brackets and apply Corollary~\ref{consresultcor2} to `$\RCAO\vdash (\forall^{\st}\Phi, \Xi, f^{1})(\exists^{\st}n)C(\Phi, \Xi, f, n)$' to obtain a term $t$ such that $\RCAo$ proves 
\be\label{drifgs}
 (\forall \Phi, \Xi, f^{1})(\exists n\in t(\Phi, \Xi, f))C(\Phi, \Xi, f, n).  
 \ee
Now define the term $s(\Phi, \Xi, f)$ as $\max_{i<|t(\Phi, \Xi, f)|}t(\Phi, \Xi, f)(i)$ and 
note that the formula $(\exists n\in t(\Phi, \Xi, f))C(\Phi, \Xi, f, n)$ implies $C(\Phi, \Xi, f, s(\Phi, \Xi, f))$.
Finally, bring the quantifier involving $f$ inside $C$ to obtain for all $\Phi, \Xi$ that
\[
 [(\forall x^{1}\in[0,1])\UBIN(\Phi, x)\wedge (\forall  U^{1}, S^{1}, k^{0})A(U, S, \Xi(U, S, k), k, \Phi)]  \di   (\forall f^{1}) B(f,s(\Phi, \Xi, f)).
\]
Thus, $s(\Phi, \Xi, \cdot)$ provides the functional $(\mu^{2})$ if $\Phi$ satisfies $(\forall x^{1}\in[0,1])\UBIN(\Phi,x )$ and $\Xi$ is the associated extensionality functional. 

\medskip

Finally, to prove $\paai\di \UBIN^{+}$, consider \eqref{frux} and apply $\HAC_{\INT}$ to the former to obtain $\nu^{1\di 0^{*}}$ such that $(\forall^{\st}f^{1})(\exists i^{0}\in \nu({f}))A(f, i)$, where $A$ is the formula in square brackets in \eqref{frux}.  Now define the \emph{standard} functional $\xi^{2}$ by 
\[
\xi(f):=(\mu m\leq \max_{i<|\nu(f)|}\nu(f)(i))(f(m)=0)
\]
and note that $[\MU(\xi)]^{\st}$, i.e.\ we have access to Feferman's search operator relative to `st'.  In particular, $\xi^{2}$ provides arithmetical comprehension (and \emph{Transfer} for $\Pi_{1}^{0}$-formulas):
\be\label{karmic}
(\forall^{\st}f^{1})\big[( f(\xi(f))=0) \asa (\exists m^{0})(f(m)=0)\asa (\exists^{\st} m^{0})(f(m)=0)  \big].  
\ee
To define $\Phi$ as in $\UBIN^{+}$, use $\xi$ from \eqref{karmic} to decide if $x\geq_{\R}\frac12$ or $x<_{\R}\frac12$ and define $\Phi(x)(0)$ as $1$ or $0$ respectively.  
Similarly, define $\Phi(x)(n+1)$ as $1$ or $0$ depending on whether $x\geq_{\R}\frac{1}{2}(x+\sum_{i=0}^{n}\frac{\Phi(x)(i)}{2^{i}})$ or not, again using $\xi$.  Then $\Phi$ is standard and 
satisfies $[\UBIN(\Phi)]^{\st}$.  Now apply $\paai$ to the latter and the axiom of extensionality to obtain $\UBIN^{+}$.  The first conjunct of \eqref{frood7} now follows in the same way as in the first part of the proof.      
\end{proof}
Note that the non-computable power of \emph{uniform} $\BIN$ (both nonstandard and non-nonstandard) arises from the fact that not all reals have a unique binary expansion.
Hence, for small (infinitesimal) variations of the input of the functional in $\UBIN$, we can produce large (standard) variations in the output.  This is exploited as follows in the previous proof.  
\begin{rem}[Nonstandard extensionality trick]\label{trick}\rm
First of all, we note that $\Phi$ as in $\UBIN(\Phi)$ is \emph{nonstandard} discontinuous in that for every $x,y\in \R$ such that $x<_{\R}\frac{1}{2}<_{\R}y \wedge x \approx \frac{1}{2}\approx y$ we have $\Phi(x)\not\approx_{1}\Phi(y)$, in particular $\Phi(x)(1)=0\ne 1=\Phi(y)(1)$.  
Secondly, we use \eqref{sly} to define \emph{standard} points $x_{\pm}$ at which $\Phi$ from $\UBIN^{+}$ is \emph{nonstandard discontinuous}, assuming $\neg\paai$.    
The ensuing contradiction with the standard extensionality of $\Phi$ yields $\UBIN^{+}\di \paai$.  Thirdly, applying term extraction to the (normal form \eqref{enoka} of the) latter implication, we obtain the effective implication \eqref{frood7}.  % of $\UBIN\di (\mu^{2})$.  
\end{rem}
The previous technique is similar in spirit to Grilliot's trick, but note that our nonstandard technique produces an effective implication, \emph{without paying attention to effective content}.  
In particular, we used the non-constructive `proof by contradiction' to establish $\UBIN^{+}\di \paai$, and `independence of premises' to obtain the latter's normal form \eqref{enoka} (See e.g.\ \eqref{frux} and \eqref{kurve2}).  

\medskip

Note that we do not claim that the previous theorem (or the below theorems) is unique or a first in this regard: Kohlenbach's treatment of Grilliot's trick (\cite{kooltje}) and the \emph{proof mining} program (\cite{kohlenbach3}) are well-known to produce effective results from classical mathematics.  
What is surprising about results in this paper (in our opinion) is the emergence of effective content (with relative ease) from Nonstandard Analysis \emph{despite} claims 
that the latter is somehow fundamentally non-constructive by e.g.\ Bishop and Connes (See \cite{samsynt} for a detailed discussion of the Bishop-Connes critique).  

\medskip

Surprisingly, the proof of Theorem \ref{proto7} goes through constructively, as we discuss now.  
\begin{rem}[The system $\H$]\label{reasonsofspace}\rm
The system $\H$ is a conservative extension of Heyting arithmetic satisfying a term extraction theorem similar to Theorem \ref{consresult} (See \cite{brie}*{Theorem~5.9}).
Although $\H$ is based on intuitionistic logic, it does prove the following `standard' version of Markov's principle (See \cite{brie}*{p.\ 1978}):
\be\tag{$\textsf{MP}^{\st}$}
(\forall^{\st} f^{1})\big[ \neg\neg[(\exists^{\st}m)(f(m)=0)]  \di (\exists^{\st}n)(f(n)=0)    \big].
\ee
Now, the proof of $ \UBIN^{+}\di \paai$ in Theorem \ref{proto7} easily yields a proof of: 
\be\label{lekkerbekske}
\UBIN^{+} \di (\forall^{\st}f^{1})\big[(\exists m)(f(m)=0) \di \neg[(\forall^{\st}n^{0})(f(n)=0)] \big].
\ee
inside the system $\H$.  
However, combined with $\textsf{MP}^{\st}$, \eqref{lekkerbekske} yields that $\H$ also proves the implication $ \UBIN^{+}\di \paai$, using the same `proof by contradiction' proof used for Theorem \ref{proto7}.   Furthermore, similar to $\textsf{MP}^{\st}$, the system $\H$ also contains a `standard' version of the independence of premise schema (in the form of $\textsf{HIP}_{\forall^{\st}}$; see \cite{brie}*{p.\ 1978}).  Thanks to this schema, $\H$ proves $\paai \di \eqref{frux}$ and even that $\paai \di \UBIN^{+}$ implies its normal form \eqref{enoka}.  Applying the term extraction theorem \cite{brie}*{Theorem~5.9} for $\H$, a constructive proof of \eqref{frood7} is established.    
\end{rem}

\subsection{Weak weak K\"onig's lemma}\label{FUWKL}
In this section, we study the principle  \emph{weak weak K\"onig's lemma} ($\WWKL$ for short), using the standard extensionality trick in Remark \ref{trick}.
%In Section~\ref{knowledge}, we study the \emph{contraposition} of $\WWKL$, with quite different results.  
Note that $\WWKL$ was not directly studied in \cites{samzoo, samzooII}.  % despite being in the RM zoo (\cite{damirzoo}).     
%ined as follows.
\bdefi[Weak weak K\"onig's lemma]\label{leipi}~
\begin{enumerate}
\item We reserve `$T^{1}$' for trees and denote by `$T^{1}\leq_{1}1$' that $T$ is a \emph{binary} tree.  
\item For a binary tree $T$, define $\nu(T):=\lim_{n\di \infty}\frac{\{\sigma \in T: |\sigma|=n    \}}{2^{n}}$.
\item For a binary tree $T$, define `$\nu(T)>_{\R}a^{1}$' as $(\exists k^{0})(\forall n^{0})\big(\frac{\{\sigma \in T: |\sigma|=n    \}}{2^{n}}\geq a+\frac{1}{k}\big)$.
\item We define $\WWKL$ as $(\forall T \leq_{1}1)\big[ \nu(T)>_{\R}0\di (\exists \beta\leq_{1}1)(\forall m)(\overline{\beta}m\in T) \big]$.
%\ee
\end{enumerate}     
\edefi
The principle $\WWKL$ is not part of the `Big Five' of RM, but there are \emph{some} equivalences involving the former (See \cite{simpson2}*{X.1}).  
In this section, we study the following uniform versions:
\be\tag{$\UWWKL$}
(\exists \Phi^{1\di 1})(\forall T \leq_{1}1)\big[ \nu(T)>_{\R}0\di (\forall m)(\overline{\Phi(T)}m\in T) \big]
\ee
Also, $\UWWKL(\Phi(T), T)$ is $\UWWKL$ without the leading quantifiers, and $\UWWKL^{+}$ is
\[
(\exists^{\st}\Phi^{1\di 1})\big[(\forall^{\st}T^{1})\UWWKL(\Phi(T), T)\wedge (\forall^{\st} T^{1}, S^{1})\big(T\approx_{1} S \di \Phi(T)\approx_{1}\Phi(S) \big)\big].
\]
Note that the second conjunct expresses that $\Phi$ is \emph{standard extensional}.  
We have the following theorem, which is the effective version of \cite{yamayamaharehare}*{Theorem 3.2}.  Note that  $\UWWKL(\Phi)$ is $(\forall T \leq_{1}1)\UWWKL(\Phi(T), T)$.  
\begin{thm}\label{proto}
From a proof in $\RCAO$ that $\UWWKL^{+}\asa \paai$, terms $s,t$ can be extracted such that $\RCAo$ proves:
\be\label{frood8}
(\forall \mu^{2})\big[\textsf{\MU}(\mu)\di \UWWKL(s(\mu)) \big] \wedge (\forall \Phi^{1\di 1})\big[ \UWWKL(\Phi)\di  \MU(t(\Phi, \Xi))  \big],
\ee
where $\Xi$ is an extensionality functional for $\Phi$.  
\end{thm}
\begin{proof}
First of all, to prove $\UWWKL^{+}\di \paai$ in $\RCAO$, assume $\UWWKL^{+}$ and suppose that $\paai$ is false, i.e.\ there is $f$ such that $(\forall^{\st}n)f(n)=0 \wedge (\exists m^{0})f(m)\ne 0$.
Now define the trees $T_{i}$ for $i=0,1$ as follows
\[
\sigma\in T_{i}\asa \big[\sigma(0)=i \vee \big[\sigma(0)=1-i\wedge (\forall m\leq |\sigma|)f(m)=0\big] \big].
\]
By the definition of $T_{i}$ and the behaviour of $f$, we have $T_{0}\approx_{1} T_{1}\approx_{1} 2^{\N}$, where the latter is the full binary tree and $\N:=\{n^{0}:n=_{0}n\}$.  Furthermore, $\nu(T_{0})=\nu(T_{1})=\frac12$ hold, and observe that $T_{0}$ (resp.\ $T_{1}$) only has paths starting with $0$ (resp.\ $1$).  Hence, we have 
$\Phi(T_{0})(0)=0\ne 1= \Phi(T_{1})(0)$ for $\Phi$ as in $\UWWKL^{+}$, which yields $\Phi(T_{0})\not\approx_{1} \Phi(T_{1})$.  Clearly, the latter contradicts the standard extensionality of $\Phi$.    
In light of this contradiction, we have $\UWWKL^{+}\di \paai$.  

\medskip

Secondly, to prove $\paai\di \UWWKL^{+}$, note that $\paai$ implies \eqref{karmic} as established in the proof of Theorem \ref{proto7}.   
To define $\Phi$ as in $\UWWKL^{+}$, let standard $T^{1}\leq_{1}1$ be such that $\nu(T)>_{\R}0$ and use (standard) $\xi$ from \eqref{karmic} to decide if
\be\label{contjas}
(\forall^{\st} n^{0})(\exists \beta^{0^{*}}\in T)(\beta(0)=1 \wedge |\beta|=n)\textup{ or } (\forall^{\st} n^{0})(\exists \beta^{0^{*}}\in T)(\beta(0)=0 \wedge |\beta|=n),
\ee
and define $\Phi(T)(0)$ as $1$ if the first formula in \eqref{contjas} holds, and $0$ otherwise.  
Similarly, for $\Phi(T)(m+1)$ again use $\xi$ from \eqref{karmic} to decide if the following formula holds:
\[
(\forall^{\st} n^{0}\geq m+1)(\exists \beta^{0^{*}}\in T)(\overline{\beta}m=\Phi(T)(0)*\dots* \Phi(T)(m)  \wedge\beta(m+1)=1 \wedge |\beta|=n) 
\]
and define $\Phi(T)(m+1)$ as $1$ if it does, and zero otherwise.  % or $0$ depending on whether we have
Then $\Phi$ is standard and satisfies $[\UWWKL(\Phi)]^{\st}$.  Now apply $\paai$ to the latter and \eqref{EXT} to obtain $\UWWKL^{+}$. 

\medskip

Thirdly, we now prove the second conjunct in \eqref{frood8}.    
Note that $\paai$ can easily be brought into the normal form \eqref{frux}
where the formula in square brackets is abbreviated by $B(f, i)$.  Similarly, the second conjunct of $\UWWKL^{+}$ has the following normal form:  
\be\label{kurve}
(\forall^{\st} T^{1}, S^{1}, k^{0})(\exists^{\st} N)\big[\overline{T}N=_{0}\overline{S}N \di \overline{\Phi(T)}k=_{0}\overline{\Phi(S)}k\big],
\ee
which is immediate by resolving `$\approx_{1}$' and bringing standard quantifiers outside. 
We denote the formula in square brackets in \eqref{kurve} by $A(T, S, N, k, \Phi)$.  
Hence, the implication $\UWWKL^+\di \paai$ now immediately yields:
\begin{align}
(\forall^{\st}\Phi, \Xi)\big[ [(\forall^{\st}T^{1})\UWWKL(\Phi(T), T)&\wedge (\forall^{\st} U^{1}, S^{1}, k^{0})A(U, S, \Xi(U, S, k), k, \Phi)]\notag\\
&  \di     (\forall^{\st}f^{1})(\exists^{\st}n)B(f,n)\big],\label{similarluuuuu}
\end{align}
by strengthening the antecedent by introducing $\Xi$.
Dropping the `st' in the antecedent of the implication and bringing out the remaining standard quantifiers:
\[
(\forall^{\st}\Phi, \Xi ,f^{1})(\exists^{\st}n)\big[ [(\forall T^{1})\UWWKL(\Phi(T), T)\wedge (\forall  U, S, k)A(U, S, \Xi(U, S, k), k, \Phi)]  \di    B(f,n)\big]. 
\]
Let $C(\Phi, \Xi, f, n)$ be the formula in big square brackets and apply Corollary~\ref{consresultcor2} to `$\RCAO\vdash (\forall^{\st}\Phi, \Xi, f^{1})(\exists^{\st}n)C(\Phi, \Xi, f, n)$' to obtain a term $t$ such that $\RCAo$ proves 
\be\label{drifgs2}
 (\forall \Phi, \Xi, f^{1})(\exists n\in t(\Phi, \Xi, f))C(\Phi, \Xi, f, n).  
 \ee
Now define the term $s(\Phi, \Xi, f)$ as $\max_{i<|t(\Phi, \Xi, f)|}t(\Phi, \Xi, f)(i)$ and 
note that the formula $(\exists n\in t(\Phi, \Xi, f))C(\Phi, \Xi, f, n)$ implies $C(\Phi, \Xi, f, s(\Phi, \Xi, f))$.
Finally, bring the quantifier involving $f$ inside $C$ to obtain for all $\Phi, \Xi$ that
\[
 [(\forall T^{1})\UWWKL(\Phi(T), T)\wedge (\forall  U^{1}, S^{1}, k^{0})A(U, S, \Xi(U, S, k), k, \Phi)]  \di   (\forall f^{1}) B(f,s(\Phi, \Xi, f)).
\]
Thus, $s(\Phi, \Xi, \cdot)$ provides $(\mu^{2})$ if $\Phi$ satisfies $(\forall T^{1})\UWWKL(\Phi(T),T )$ and $\Xi$ is the associated extensionality functional. 

\medskip

Finally, the first conjunct in \eqref{frood8} is proved as follows: $\paai\di \UWWKL^{+}$ yields
\be\label{kanttt}
(\forall^{\st}f^{1})(\exists^{\st}n)B(f,n)\di  (\forall^{\st}T^{1})(\exists^{\st}\alpha^{1}\leq_{1}1 )\UWWKL(\alpha, T),
\ee
where we used the same notations $ B$ as in \eqref{similarluuuuu}.  Since standard functionals yield standard output for standard input by the basic axioms of Definition \ref{debs}, \eqref{kanttt} yields
\be\label{kanttt2}
(\forall^{\st} \mu^{2})\big[(\forall^{\st}f^{1})B(f,\mu(f))\di  (\forall^{\st}T^{1})(\exists^{\st}\alpha^{1}\leq_{1}1 )\UWWKL(\alpha, T)\big].
\ee
Weakening the antecedent of \eqref{kanttt2} and bringing outside the standard quantifiers, we obtain
\be\label{kanttt3}
(\forall^{\st} \mu^{2}, T^{1})(\exists^{\st}\alpha^{1}\leq_{1}1 )\big[(\forall f^{1})B(f,\mu(f))\di  \UWWKL(\alpha, T)\big].
\ee
Applying Corollary \ref{consresultcor2} to `$\RCAO\vdash \eqref{kanttt3}$', we obtain a term $t$ such that $\RCAo$ proves 
\be\label{bl;og}
(\forall \mu^{2}, T^{1})(\exists \alpha^{1}\in t(\mu, T) )\big[(\forall f^{1})B(f,\mu(f))\di  \UWWKL(\alpha, T)\big].
\ee
Since $\mu^{2}$ satisfying the antecedent of \eqref{bl;og} is indeed Feferman's mu as in $(\mu^{2})$, we can select the `right' $\alpha\in t(\mu, T)$, and we have obtained the first conjunct in \eqref{frood8}.
\end{proof}
As suggested by its name, $\WWKL$ is a weakening of $\WKL$, namely to binary trees with positive measure.  
Hence, the above proof should also go through for $\UWKL^{+}$, which is $\UWWKL^{+}$ \emph{without} the aforementioned restriction.  
In particular, the measure does not play any role except saying that the tree is infinite.
We hasten to add that the restriction to trees of positive measure \emph{does} play an important role for $\STP$ and $\LMP$ by Theorem \ref{komo}. 
\subsection{The intermediate value theorem}\label{hivt}
In this section, we study the \emph{intermediate value theorem} \textsf{IVT}.   
%we establish a particular nonstandard equivalence involving $\paai$ and a (nonstandard and uniform) version of $\IVT$, 
We will apply the nonstandard extensionality trick from Remark~\ref{trick} to a nonstandard version 
of the `usual' Brouwerian counterexample to $\IVT$ from \cite{beeson1}*{I.7}.

\medskip

Now, $\IVT$ has a proof in $\RCA_{0}$ when formulated in second-order arithmetic using so-called RM codes (\cite{simpson2}*{II.6.2}).  
However, $\IVT$ is not constructively true (See \cite{beeson1}*{I.7}) as the aforementioned proof makes essential use of the law of excluded middle, while a finer classification may be found in \cite{bergske}.  As a consequence of this non-constructive status, there is an equivalence (See \cite{kohlenbach2}*{Prop.~3.14}) between the Turing jump functional $(\exists^{2})$ and uniform $\IVT$, the latter defined as $\UIVT$ as follows (See \cite{kohlenbach2}*{\S3} for a number of variations):
\be\tag{$\UIVT$}
(\exists \Phi^{(1\di1)\di1 })(\forall f\in \overline{C})[f(\Phi(f))=_{\R}0], 
\ee
where `$f\in \overline{C}$' is short for `$f(0)<_{\R}0 <_{\R} f(1)\wedge \eqref{kraakje}$' where the latter is as follows
\be\label{kraakje}\textstyle
(\forall k^{0})(\forall x^{1}\in [0,1])(\exists N^{0})(\forall y^{1}\in [0,1])( |x-y|<_{\R}\frac{1}{N}\di |f(x)-f(y)|<_{\R}\frac1k), 
\ee
i.e.\ the \emph{internal} `epsilon-delta' definition of (pointwise) continuity on $[0,1]$.  We also write `$f\in C([0,1])$' if $f:\R\di \R$ satisfies \eqref{kraakje}.  
We define $\UIVT(\Phi)\equiv(\forall f\in \overline{C})[f(\Phi(f))=_{\R}0]$.   
Now consider the following nonstandard and uniform version of $\UIVT$:  % which we call $\UIVT^{+}$, 
%\be\label{ubinplus}\tag{$\UBIN^{+}$}
\be\tag{$\UIVT^{+}$}
(\exists^{\st} \Phi)\big[(\forall^{\st} f\in \overline{C})[f(\Phi(f))=_{\R}0 ]\wedge (\forall^{\st}f, g\in \overline{C})\big( f\approx g\di \Phi(f)\approx \Phi(g)    \big) \big],
\ee
 where `$f\approx g$' is short for $(\forall^{\st}q^{0}\in [0,1])(f(q)\approx g(q))$.  %where $\UBIN(\Phi, x)$ is the formula in square brackets in $\UBIN$.  
Note that the second conjunct of $\UIVT^{+}$ expresses that $\Phi$ is `standard extensional', i.e.\ satisfies extensionality as in \eqref{furg} relative to `st'.  
We have the following theorem.  
\begin{thm}\label{proto667}
From a proof in $\RCAO$ that $\UIVT^{+}\asa \paai$, terms $s,t$ can be extracted such that $\RCAo$ proves:
\be\label{frood667}
(\forall \mu^{2})\big[\textsf{\MU}(\mu)\di \UIVT(s(\mu)) \big] \wedge (\forall \Phi^{1\di 1})\big[ \UIVT(\Phi)\di  \MU(t(\Phi, \Xi))  \big],
\ee
where $\Xi$ is an extensionality functional for $\Phi$.
\end{thm}
\begin{proof}
First of all, we prove $\UIVT^{+}\di \paai$, for which we make use of the `usual' Brouwerian counterexample to $\IVT$ (See e.g.\ \cite{beeson1}*{Fig.\ 2, p.\ 12}).  
Let $f_{0}:\R\di \R$ be the function which is $3x-1$ for $x\in [0,\frac13]$, $3x-2$ for $x\in [\frac23, 1]$, and zero for $x\in [\frac13, \frac23]$.     
Suppose $\paai$ is false and let standard $g^{1}$ be such that $(\forall^{\st}n)(g(n)=0)$ and $(\exists m_{0})g(m_{0})\ne 0$.  Define standard functions $f_{\pm}(x):= f_{0}(x)\pm\sum_{i=0}^{\infty}\frac{g(i)}{2^{i}}$ and note that $f_{+}(x)\approx f_{-}(x)$ for all $x\in [0,1]$.  However, for $\Phi$ as in $\UIVT^{+}$, we have $\Phi(f_{+})<_{\R}\frac13$ and $\Phi(f_{-})>_{\R}\frac23$, which contradicts the second conjunct of $\UIVT^{+}$.  Hence, $\paai$ follows and we obtain $\UIVT^{+}\di \paai$.

\medskip
  
Secondly, for the reversal $\paai \di \UIVT^{+}$, fix standard $k^{0}, x^{1}\in [0,1] $ and $f:\R\di\R$ in \eqref{kraakje} and consider the following $\Sigma_{2}^{1}$-formula:
\be\label{boal}\textstyle
(\exists N^{0})(\forall q^{0}\in [0,1])( |x-q|<_{\R}\frac{1}{N}\di |f(x)-f(q)|\leq_{\R}\frac1k).
\ee
Applying\footnote{
To obtain $\eqref{boal}\di\eqref{aimino}$, note that $\neg\eqref{aimino}$ has a normal form; apply $\HAC_{\INT}$ to remove the existential quantifier and then apply $\paai$ to the resulting formula to obtain a contradiction with \eqref{boal}.
In general, similar to how one `bootstraps' $\Pi_{1}^{0}$-comprehension from $\ACA_{0}$, the system $\RCAO+\paai$ proves $\varphi\asa \varphi^{\st}$ for \emph{any} internal arithmetical formula (only involving standard parameters).  Indeed, recall that $\paai\di (\mu^{2})^{\st}$ (See the proof of Theorem \ref{proto7}) and use the latter to remove all existential quantifiers (except the leading one, if such there is) from $\varphi^{\st}$; applying $\paai$ to the resulting universal formula yields $\varphi$.  The implication $\varphi\di \varphi^{\st}$ follows from the previous, $ \neg\varphi^{\st}\di \neg\varphi$ in particular.} $\paai$ to the formula \eqref{boal} yields:
\be\label{aimino}\textstyle
(\exists^{\st} N^{0})(\forall^{\st} q^{0}\in [0,1])( |x-q|\ll \frac{1}{N}\di |f(x)-f(q)|\lessapprox \frac1k).
\ee
Applying $\paai$ to the universal formula in \eqref{aimino}, we obtain for standard $f$ that
\be\label{opa}\textstyle
(\forall^{\st}x\in [0,1], k^{0})(\exists^{\st} N^{0})(\forall  q^{0}\in [0,1])( |x-q|<_{\R} \frac{1}{N}\di |f(x)-f(q)|\leq_{\R} \frac1k),
\ee
and applying $\HAC_{\INT}$ to the previous formula yields standard $\Phi^{(1\times 0)\di 0^{*}}$ such that 
\[\textstyle
(\forall^{\st}x\in [0,1], k^{0})(\exists  N^{0}\in \Phi(x, k))(\forall  q^{0}\in [0,1])( |x-q|<_{\R} \frac{1}{N}\di |f(x)-f(q)|\leq_{\R} \frac1k),
\]
Define $\Psi(x, k)$ as $\max_{i<|\Phi(x, k)|}\Phi(x, k)(i)$ and note that $\Psi$ provides a kind of \emph{modulus of continuity} for (standard) $f$.  
With this continuity in place, we can follow (a variation of) the classical proof of $\IVT$ in \cite{simpson2}*{II.6} to define {standard} $\Phi$ as in $\UIVT^{+}$ as follows.

\medskip

To this end, recall that $\paai$ implies the existence of standard $\xi$ as in \eqref{karmic}, which allows us to decide for standard $f\in \overline{C}$ if $(\exists^{\st}q^{0}\in [0,1])(f(q)\approx 0)$ holds or not.  
If the latter holds, use $\xi$ to find such $q^{0}$ and define $\Phi(f):=q$, which also yields $f(\Phi(f))=_{\R}0$ by $\paai$.   
If however $(\forall^{\st}q^{0}\in [0,1])(f(q)\not\approx 0)$, define $\Phi(f)(0)$ as $0$ or $\frac{1}{2}$ depending on whether $f(\frac12)\ll 0$ or $f(\frac12)\gg 0$ ($\xi$ as in \eqref{karmic} decides which disjunct holds).       
Similarly, define $\Phi(f)(n+1)$ as $\Phi(f)(n)$ or $\Phi(f)(n)+\frac{1}{2^{n+2}}$ depending on $f(\Phi(f)(n)+\frac{1}{2^{n+2}})\ll0$ or $f(\Phi(f)(n)+\frac{1}{2^{n+2}})\gg0$ ($\xi$ as in \eqref{karmic} again decides which disjunct holds).  

\medskip

Note that by $\paai$, we have $f(z)\gg 0 \asa f(z)>_{\R}0$ for any standard $z^{1}\in \R $ and $f:\R\di \R$.  
In light of \eqref{karmic}, we also have access to $I\Sigma_{1}$ relative to `st', assuming all parameters involved are standard.    
Hence, it is easy to prove that $\Phi(f)$ is a real number such that $(\forall^{\st}n)(\Phi(f)(n)<_{\R}\Phi(f)<_{\R}\Phi(f)(n)+\frac{1}{2^{n+1}})$ and $f(\Phi(f)(n))\ll 0$ and $f(\Phi(f)(n)+\frac{1}{2^{n+1}})\gg 0$.  In light of these facts and the continuity of $f$ as in \eqref{opa} for $x=\Phi(f)$, we obtain $f(\Phi(f))\approx 0$ and also $f(\Phi(f))=_{\R}0$ by $\paai$.  
Furthermore, $\Phi$ only invokes $f$ on rational numbers and $(\forall^{\st} q^{0}\in [0,1])(f(q)\approx g(q))$ for standard $f, g\in \overline{C}$ thus implies $(\forall q^{0}\in [0,1])(f(q)=_{\R} g(q))$ by $\paai$.  By the previous property and the extensionality of $f,g$, we have $\Phi(f)=_{\R}\Phi(g)$, and hence the second conjunct of $\UIVT^{+}$.     
  
\medskip

Thirdly, one obtains a normal form for $\UIVT^{+}\asa \paai$ in the same way as in the proofs of Theorems \ref{proto7} and \ref{proto}.  
Applying term extraction as in Theorem~\ref{consresult}, one then readily obtains \eqref{frood667}.  For completeness, we mention the two normal forms corresponding to $\paai\di \UIVT^{+}$ and $\UIVT^{+}\di \paai$, namely as follows:
\[
(\forall^{\st} f\in \overline{C},\mu^{2})(\exists^{\st}x\in [0,1])  \big[ (\forall g^{1})B(g,\mu(g))\di (f(x)=_{\R}0 )\big].
\]
\begin{align*}\textstyle
(\forall^{\st} \Phi, \Xi, g^{1})(\exists^{\st}n^{0})\big[\big((\forall z\in  \overline{C})&(z(\Phi(z))=_{\R}0 )\textstyle\wedge (\forall x, y \in \overline{C}, k^{0})(\exists q^{0},n^{0}\in \Xi(x,y,k))\\
&\textstyle\big( |x(q) -y(q)|<\frac{1}{n}\di |\Phi(x)- \Phi(y)|<\frac{1}{k}    \big)\big) \di B(g, n)\big],
\end{align*}
where we used the notation $B$ from \eqref{similarluuuuu}. 
To be absolutely clear, we point out that the formula `$(\forall^{\st} f\in \overline{C})A(f)$' is short for $(\forall f)([\st(f)\wedge f\in \overline{C}]\di A(f))$.  
In particular, `$f\in \overline{C}$' is the \emph{internal} formula `$f(1)>_{\R}0>_{\R} f(0)\wedge \eqref{kraakje}$' which remains \emph{untouched} by the addition of `st' to the quantifier over the variable $f$.    
For this reason, the formula `$f\in \overline{C}$' does not play any role in obtaining the normal form of $\UIVT^{+}\asa \paai$ and the associated term extraction via Theorem \ref{consresult}.  
\end{proof}
The previous proof is easily adapted to the Brouwerian counterexample to the Weierstra\ss~maximum theorem, as follows.  
\begin{rem}[Weierstra\ss~maximum theorem]\rm
The Brouwerian counterexample to the Weierstra\ss~maximum theorem is given in \cite{beeson1}*{I.6} by a function with two relative maxima.  
For instance, one can use $f_{\pm}:\R\di \R$ which is $(1\pm x_{0})|\sin 7x |$ for $x\in [0, \frac{1}{7}\pi]$ and $(1\mp x_{0})|\sin 7x |$ for $x\in [\frac{1}{7}\pi, 1]$, where $x_{0}:=\sum_{i=0}^{\infty}\frac{g(i)}{2^{i}}$ and $g$ as in the previous proof.  A functional witnessing the Weierstra\ss~maximum theorem will map $f_{+}$ to $\frac{\pi}{14}$ and $f_{-}$ to $\frac{3\pi}{14}$ while $f_{+}\approx f_{-}$ under the same assumptions as in the previous proof.
\end{rem}

\subsection{Rational numbers and the law of excluded middle}\label{carmichael}
In this section, we study the classical dichotomy that every real number is either rational or not, i.e.\  
\be\label{LEM2}\tag{\textsf{DQ}}
(\forall x\in \R)\big[ (\exists q\in \Q)(q=_{\R}x)\vee (\forall r\in \Q)(r\ne_{\R}x) \big].  
\ee
Our study will yield interesting insights into the role of the law of excluded middle in Nonstandard Analysis.  
In particular, we will offer a partial explanation \emph{why} Nonstandard Analysis can produce computational information, as observed in the previous sections (and the aforementioned references).  

\medskip

First of all, \ref{LEM2} is a trivial consequence of the law of excluded middle, and the former is indeed equivalent to \textsf{LPO} in constructive mathematics (\cite{brich}*{p.\ 5}).  
We will study the \emph{nonstandard} version of \ref{LEM2}, defined as follows:
\be\label{LEM1}\tag{$\textup{\textsf{DQ}}_{\ns}$}
(\forall^{\st} x\in \R)\big[ (\exists^{\st}q\in \Q)(q=_{\R}x)\vee (\forall r\in \Q)(r\ne_{\R}x) \big],
\ee
The uniform version of \ref{LEM2} is also obvious:
\be\label{LEM3}\tag{\textsf{UDQ}}
(\exists \Phi^{2})(\forall x\in \R)\big[ (\Phi(x)\in \Q \wedge q=_{\R}x)\vee (\forall r\in \Q)(r\ne_{\R}x) \big].  
\ee
Let $\textsf{UDQ}(\Phi)$ be the previous with the leading quantifier dropped.
\begin{thm}\label{krefje}
In $\RCAO$, we have $\paai\asa \textsf{\textup{DQ}}_{\ns}$.  From the latter proof, we can extract terms $s$ and $t$ such that
\be\label{unik}
(\forall \mu^{2})\big[\textsf{\MU}(\mu)\di \UDQ(s(\mu)) \big] \wedge (\forall \Phi^{2})\big[ \UDQ(\Phi)\di  \MU(t(\Phi))  \big].
\ee
\end{thm}
\begin{proof}
The first forward implication is trivial while the first reverse equivalence follows easily:  Suppose $\paai$ is false and consider $f^{1}$ such that $(\forall^{\st}n)(f(n)=0)\wedge (\exists m)(f(m)\ne 0)$.
Define the real $x:=\sum_{i=0}^{\infty}\frac{h(n)}{2^{n}}$ where $h(n)=1$ if $(\forall i\leq n)(f(n)=0)$, and zero otherwise.  Note that $x$ is a rational number, namely $x=_{\R} \sum_{i=0}^{m_{0}}\frac{h(n)}{2^{n}}  $, where $m_{0}$ is the last $m$ such that $f(m)\ne0$.  However, we also have $x\ne_{\R} q$ for every standard rational, and this contradiction yields $\paai$.    

\medskip

Since $\paai$ has a normal form \eqref{frux} and $\textsf{DQ}_{\ns}$ has an obvious normal form, \eqref{unik} follows in the same way as in the second part of the proof of Theorem \ref{proto7}.
\end{proof}
Secondly, we discuss an apparent (but not actual) contradiction regarding $\textsf{DQ}_{\ns}$ and the previous theorem, as follows.
\begin{rem}[Equality in $\P$]\label{druj}\rm
First of all, $\P$ proves \ref{LEM2} via the law of excluded middle.  
Hence for a \emph{standard} real $x$ which does not satisfy the second disjunct of \ref{LEM2}, we may conclude the existence of a rational $q^{0}$ such that $x$ equals $q$.  
Secondly, in light of the first basic axiom of $\P$ (See item (\ref{komit}) of Definition~\ref{debs}), $q$ must be standard as $x$ is standard, and $x$ equals $q$.  
However, this means that $\RCAO$ proves $\textsf{DQ}_{\ns}$, which is impossible in light of Theorem \ref{krefje}.  
Thirdly, this apparent contradiction is easily explained by noting that `$=_{\R}$' as defined in $\P$ (See Definition \ref{keepinitreal}) does not fall under item (\ref{komit}) of Definition \ref{debs}.  
\end{rem}
Thirdly, while Theorem \ref{krefje} is not particularly deep, this theorem inspires Remark \ref{druj}, which in turn gives rise to the following observation:  
In $\RCAo$, either a real is rational or not because of the law of excluded middle \ref{LEM2}.  
By contrast, in $\RCAO$, there are three possibilities for every standard real:
\begin{enumerate}
\item $x$ is a standard rational;
\item $x$ is not a rational;  
\item $x$ is a rational, but not standard;\label{trunki}
\end{enumerate}
and the third possibility \eqref{trunki} only disappears given $\paai$ by Theorem \ref{krefje}.  Similarly, again over $\RCAO$, for a standard function $f^{1}$, there are three possibilities:
\begin{enumerate}
\renewcommand{\theenumi}{\roman{enumi}}
\item there is standard $n^{0}$ such that $f(n)=0$;\label{work}
\item for all $m^{0}$ we have $f(m)\ne0$;
\item for all standard $n^{0}$ we have $f(n)\ne0$ while there is $m^{0}$ with $f(m)= 0$;\label{kraft}
\end{enumerate}
and the third possibility again only disappears given $\paai$.  In particular, \emph{in the extended language of $\RCAO$}, $\paai$ (and not \ref{LEM2}) is the principle which excludes the third option \eqref{trunki} and \eqref{kraft}.  In other words, it seems that $\paai$ plays the role of the law of excluded middle/third \emph{in the extended language of $\RCAO$}.        

\medskip

% by Theorem \ref{krefje}.
Furthermore, it has been suggested that the predicate `$\st(n^{0})$' can be read as `$n^{0}$ is computationally relevant' or `$n^{0}$ is calculable' in \cite{brie}*{p.\ 1963}, \cite{benno2}*{\S4}, and \cite{sambon}*{\S3.4}. 
If we read the previous items \eqref{work}-\eqref{kraft} through this filter, they reflect three well-known possibilities suggested by the BHK interpretation (See \cite{troeleke1}*{\S3.1}): 
\begin{enumerate}
\renewcommand{\theenumi}{\Roman{enumi}}
\item we can compute $n^{0}$ s.t.\ $f(n)=0$ (constructive existence);\label{kracht}
\item for all $m^{0}$ we have $f(m)\ne0$;
\item $\neg[(\forall  m^{0})(f(m) \ne0)]$;\label{werk} 
\end{enumerate}
In conclusion, we have observed that the role of the law of excluded middle \emph{in the extended language of $\RCAO$} is played by Nelson's axiom \emph{Transfer}, which is however absent from $\RCAO$.  Due to this absence, there are \emph{three} possibilities as in items \eqref{work}-\eqref{kraft}, similar to the possibilities in items \eqref{kracht}-\eqref{werk} in constructive mathematics.  
In other words, the systems $\P$ and $\RCAO$ are constructive in that they lack \emph{Transfer}, which is the law of excluded middle \emph{for the extended language of internal set theory}.  
We believe this to be a partial explanation of the vast computational content of Nonstandard Analysis established above and in \cites{sambon, samGH, samzoo, samzooII}.

\section{A new class of functionals}\label{knowledge}
We discuss the results from \cite{dagsam}, some of which were announced in \cite{samGT2}, and related results. 
The associated connection between Nonstandard Analysis and computability theory forms the motivation for this paper, as discussed in Section \ref{intro}.  

\subsection{Nonstandard Analysis and Computability Theory: an introduction}\label{connexis}
The connection between \emph{computability theory} and \emph{Nonstandard Analysis} is investigated \cite{dagsam}.  The two following topics are investigated \emph{and} shown to be intimately related.  
\begin{enumerate}[label=\textsf{(T.\arabic*)}]
\item[\textsf{(T.1)}] A basic property of \emph{Cantor space} $2^{\N}$ is \emph{Heine-Borel compactness}: For any open cover of $2^{\N}$, there is a \emph{finite} sub-cover.   
A natural question is: \emph{How hard is it to {compute} such a finite sub-cover}?  This is made precise in \cite{dagsam} by analysing the complexity of functionals that for $g:2^{\N}\di \N$, 
output a finite sequence $\langle f_0 , \dots, f_n\rangle $ in $2^{\N}$ such that the neighbourhoods defined from $\overline{f_i}g(f_i)$ for $i\leq n$ form a cover of Cantor space.
\item[\textsf{(T.2)}] A basic property of Cantor space in \emph{Nonstandard Analysis} is Abraham Robinson's \emph{nonstandard compactness} (See \cite{loeb1}*{p.\ 42}), i.e.\ that every binary sequence is `infinitely close' to a \emph{standard} binary sequence.  The strength of this nonstandard compactness property of Cantor space is analysed in \cite{dagsam} and compared to the other axioms of Nonstandard Analysis and usual mathematics.
\end{enumerate}
The study of \textsf{(T.1)} in \cite{dagsam} involves the \emph{special fan functional} $\Theta$, discussed in Section \ref{thespecial} below and first introduced in \cite{samGH}.  
Clearly, Tait's \emph{fan functional} (\cite{noortje}) computes\footnote{Tait's fan functional $\Phi$ computes a modulus of \emph{uniform} continuity $N^{0}=\Phi(g)$ for any continuous functional $g:2^{\N}\di \N$.  The modulus $N^{0}$ yields a supremum for $g$ by computing the maximum of $g(\sigma*00\dots)$ for all binary sequences $\sigma$ of length $N$.} a sequence $\langle f_0 , \dots, f_n\rangle $ as in \textsf{(T.1)} for \emph{continuous} $g:2^{\N}\di \N$, while the special fan functional does so \emph{for any}  $g:2^{\N}\di \N$.  This generalisation from continuous to general inputs is interesting (and even necessary) in our opinion as mathematics restricted to e.g.\ only recursive objects, like the Russian school of recursive mathematics, can be strange and counter-intuitive (See \cite{beeson1}*{Chapter IV} for this opinion).  
Some of the (highly surprising) computational properties of $\Theta$ established in \cite{dagsam} are discussed in Section \ref{thespecial}.  In particular, $\Theta$ seems extremely hard to compute (as in Kleene's S1-S9 from \cite{longmann}*{\S5.1}) as no type two functional can compute it.  

\medskip

The study of \textsf{(T.2)} in \cite{dagsam} amounts to developing the Reverse Mathematics of Nonstandard Analysis.  For instance, the nonstandard counterparts of $\WKL_{0}$ and $\WWKL_{0}$ are $\STP$ and $\LMP$ (See Section \ref{thespecial} for definitions), 
each expressing a nonstandard kind of compactness.  On the other hand, the nonstandard counterpart of $\ACA_{0}$ is $\paai$ introduced above.  While we have $\ACA_{0}\di \WKL_{0}\di \WWKL_{0}$ in RM, the nonstandard counterparts behave quite differently, namely 
we have $\paai \not\di \STP$ and $\paai \not\di\LMP$, and much stronger non-implications involving the nonstandard counterpart of $\FIVE$, the strongest `Big Five' system.  

\medskip

We stress that \textsf{(T.1)} and \textsf{(T.2)} are highly intertwined and that the study of these topics in \cite{dagsam} is `holistic' in nature: 
results in computability theory give rise to results in Nonstandard Analysis \emph{and vice versa}, as discussed in the next section.   By way of a basic example, consider $\Theta$ as in \textsf{(T.1)} and recall that the output of $\Theta$ is readily computed in terms of Tait's fan functional if $g^{2}$ is continuous on Cantor space.  Experience bears out that the uninitiated express extreme skepticism about the fact that $\Theta$ is also well-defined for \emph{discontinuous} inputs $g^{2}$.  However, $\Theta$ almost trivially emerges form the nonstandard compactness of Cantor space, i.e.\ Nonstandard Analysis tells us that the special fan functional $\Theta$ exists and is well-defined.  Furthermore, the fact that the Turing jump functional from $(\exists^{2})$ cannot compute $\Theta$ as mentioned in \textsf{(T.1)} readily implies the non-implication $\paai\not \di \STP$ from \textsf{(T.2)}.    
More examples are discussed in the next section, and of course \cite{dagsam}.  

\medskip

Finally, we have sketched a connection between Nonstandard Analysis and computability theory.  However, the better part of the latter does not obviously have a counterpart in the former and vice versa.  
We list two examples:  First of all, the fact that no type two functional computes $\Theta$ is proved in computability theory (See \cite{dagsam}*{\S3}) using \emph{Gandy selection} (See \cite{longmann}*{Theorem 5.4.5}), but what is the nonstandard counterpart of the latter theorem?  
Secondly, the \emph{Loeb measure} is one of the crown jewels of Nonstandard Analysis (\cite{loeb1}), but what is the computability theoretic counterpart of this measure?  Note that first steps in this direction have been taken in \cite{samnewarix}.  
In the paper at hand, we have formulated a nonstandard counterpart of \emph{Grilliot's trick} inspired by the above connection between Nonstandard Analysis and computability theory.  

\subsection{The special fan functional and related topics}\label{thespecial}
We introduce the \emph{special fan functional} and discuss how it derives from the \emph{Standard Part} axiom of Nonstandard Analysis and why it does not belong to any existing category in RM.  

\medskip

Our motivation for this study is the following discrepancy:  On one hand, there is literally a `zoo' of theorems in RM (\cite{damirzoo}) which do fit into the `Big Five classification' of RM.  On the other hand, as shown above and in \cites{samzoo, samzooII}, uniform theorems are mostly equivalent to $(\exists^{2})$, with some exceptions based on the contraposition of $\WKL$, i.e.\ the fan theorem.  Thus, the `higher-order RM zoo' consisting of uniform theorems is still rather sparse compared to the original RM zoo.      
In this light, it is a natural question whether the higher-order RM zoo can be made as populous as the original RM zoo.  % from \cite{damirzoo}.       

\medskip

As a first step towards an answer to the aforementioned question, we discuss the following functional.  Note that $1^{*}$ is the type of finite sequences of type $1$.
\bdefi[Special fan functional]\label{special}
We define $\SCF(\Theta)$ as follows for $\Theta^{(2\di (0\times 1^{*}))}$:
\[
(\forall g^{2}, T^{1}\leq_{1}1)\big[(\forall \alpha \in \Theta(g)(2))  (\overline{\alpha}g(\alpha)\not\in T)
\di(\forall \beta\leq_{1}1)(\exists i\leq \Theta(g)(1))(\overline{\beta}i\not\in T) \big]. 
\]
Any functional $\Theta$ satisfying $\SCF(\Theta)$ is referred to as a \emph{special fan functional}.
\edefi
As noted in \cite{dagsam} and above, from a computability theoretic perspective, the main property of the special fan functional $\Theta$ is the selection of $\Theta(g)(2)$ as a finite sequence of binary sequences $\langle f_0 , \dots, f_n\rangle $ such that the neighbourhoods defined from $\overline{f_i}g(f_i)$ for $i\leq n$ form a cover of Cantor space;  almost as a by-product, $\Theta(g)(1)$ can then be chosen to be the maximal value of $g(f_i) + 1$ for $i\leq n$. 
We stress that $g^{2}$ in $\SCF(\Theta)$ may be \emph{discontinuous} and that Kohlenbach has argued for the study of discontinuous functionals in RM (\cite{kohlenbach2}).

\medskip

The name of $\Theta$ from the previous definition is due to the fact that a special fan functional may be computed from the intuitionistic fan functional $\Omega^{3}$, as in Theorem \ref{kinkel}.
%\bdefi[Intuitionistic fan functional]
\be\tag{$\MUC(\Omega)$}
(\forall Y^{2}) (\forall f, g\leq_{1}1)(\overline{f}\Omega(Y)=\overline{g}\Omega(Y)\notag \di Y(f)=Y(g)).   % \label{lukl3}\tag{$\textsf{\textup{MUC}}(\Phi)$}$
\ee
%\edefi
As to the logical strength of $(\exists \Omega)\MUC(\Omega)$, the latter gives 
rise to a conservative extension of the system $\WKL_{0}$ by \cite{kohlenbach2}*{Prop.\ 3.15}.  
\begin{thm}\label{kinkel}
There is a term $t$ such that $\textsf{\textup{E-PA}}^{\omega}$ proves $(\forall \Omega^{3})(\MUC(\Omega)\di \SCF(t(\Omega))) $.   
\end{thm}
\begin{proof}
The theorem was first proved \emph{indirectly} in \cite{samGH}*{\S3} by applying Theorem~\ref{consresult} to a suitable nonstandard implication.    
For completeness, we include the following direct proof which can also be found in \cite{dagsam}. 
Note that $\Theta(g)$ as in $\SCF(\Theta)$ has to provide a natural number and a finite sequence of binary sequences.
The number $\Theta(g)(1)$ is defined as $\max_{|\sigma|=\Omega(g)\wedge \sigma\leq_{0^{*}}1}g(\sigma*00\dots)$ and the finite sequence of binary sequences $\Theta(g)(2)$
consists of all $\tau*00\dots$ where $|\tau|=\Theta(g)(1)\wedge \tau\leq_{0^{*}}1$.  
We have for all $g^{2}$ and $T^{1}\leq_{1}1$:
\be\label{difffff}
 (\forall \beta\leq_{1}1)(\beta \in \Theta(g)(2)\di \overline{\beta}{g}(\beta)\not \in T)\di (\forall \gamma\leq_{1}1)(\exists i\leq \Theta(g)(1))(\overline{\gamma}i\not \in T).
\ee
Indeed, suppose the antecendent of \eqref{difffff} holds.  Now take $\gamma_{0}\leq_{1}1$, and note that $\beta_{0}=\overline{\gamma_{0}}\Theta(g)(1)*00\dots \in \Theta(g)(2)$, implying 
$\overline{\beta_{0}}{g}(\beta_{0})\not \in T$.  But $g(\alpha)\leq \Theta(g)(1)$ for all $\alpha\leq_{1}1$, by the definition of $\Omega$, implying that $\overline{\gamma_{0}}{g}(\beta_{0})=\overline{\beta_{0}}{g}(\beta_{0})\not \in T$ by the definition of $\beta_{0}$, and the consequent of \eqref{difffff} follows.    
\end{proof}
In light of the previous, $\Theta$ exists at the level of $\WKL_{0}$ and it therefore stands to reason that it would be \emph{easy} to compute.  
We have the following surprising theorem where `computable' should be once again interpreted in the sense of Kleene's S1-S9 (See \cite{longmann}*{5.1.1}).
The metatheory is -as always- $\ZFC$ set theory.  
\begin{thm}\label{import}
Let $\varphi^{2}$ be any functional of type two.  
Any functional $\Theta^{3}$ as in $\SCF(\Theta)$ is not computable in $\varphi$.   
Any functional $\Theta^{3}$ as in $\SCF(\Theta)$ is computable in $(\exists^{3})$ as follows
\be\tag{$\exists^{3}$}
(\exists E_{3})(\forall \varphi^{2})\big[ (\exists f^{1})(\varphi(f)=0)\asa E_{3}(\varphi)=0   \big].
\ee
\end{thm}
\begin{proof}
A proof may be found in \cite{dagsam}.
\end{proof}
By the previous, $\Theta$ is quite different from the usual\footnote{Note that the connection between $\Theta$ and Kohlenbach's generalisations of $\WKL$ from \cite{kohlenbach4}*{\S5-6} is discussed in \cite{dagsam}*{\S4}.  This connection turns out to be quite non-trivial and interesting.} objects studied in (higher-order) RM.  
An obvious question is: Where does the special fan functional and its behaviour come from?  
The answer is as follows: The nonstandard counterpart of $\WKL_{0}$ is defined as:
\be\tag{$\STP$}
(\forall \alpha^{1}\leq_{1}1)(\exists^{\st}\beta^{1}\leq_{1}1)(\alpha\approx_{1}\beta), 
\ee  
which has the following normal form, already reminiscent of $\Theta$.  
\begin{thm}\label{lapdog}
In $\P$, $\STP$ is equivalent to the following normal form:
\begin{align}\label{frukkklk}
(\forall^{\st}g^{2})(\exists^{\st}w^{1^{*}})\big[(\forall T^{1}\leq_{1}1)(\exists ( \alpha^{1}\leq_{1}1,  &~k^{0}) \in w)\big((\overline{\alpha}g(\alpha)\not\in T)\\
&\di(\forall \beta\leq_{1}1)(\exists i\leq k)(\overline{\beta}i\not\in T) \big)\big]. \notag
\end{align}  
The system $\P+(\exists^{\st}\Theta)\SCF(\Theta)$ proves $\STP$.  
\end{thm}
\begin{proof}  
The following proof is implicit in the results in \cite{samGH}*{\S3} and is added for completeness.  
First of all, $\STP$ is easily seen to be equivalent to 
\begin{align}\label{fanns}
(\forall T^{1}\leq_{1}1)\big[(\forall^{\st}n)(\exists \beta^{0})&(|\beta|=n \wedge \beta\in T ) \di (\exists^{\st}\alpha^{1}\leq_{1}1)(\forall^{\st}n^{0})(\overline{\alpha}n\in T)   \big],
\end{align}
and this equivalence may also be found in \cite{samGH}*{Theorem 3.2}.
For \eqref{frukkklk}$\di$\eqref{fanns}, note that \eqref{frukkklk} implies for all standard $g^{2}$
\begin{align}\label{frukkklk2}
(\forall T^{1}\leq_{1}1)(\exists^{\st} ( \alpha^{1}\leq_{1}1,  &~k^{0})\big[(\overline{\alpha}g(\alpha)\not\in T)
\di(\forall \beta\leq_{1}1)(\exists i\leq k)(\overline{\beta}i\not\in T) \big], 
\end{align}  
which in turn yields, by bringing all standard quantifiers inside again, that:
\begin{align}\label{frukkklk3}
(\forall T\leq_{1}1) \big[(\exists^{\st}g^{2})(\forall^{\st}\alpha \leq_{1}1)(\overline{\alpha}g(\alpha)\not\in T)\di(\exists^{\st}k)(\forall \beta\leq_{1}1)(\overline{\beta}k\not\in T) \big], 
\end{align}  
To obtain \eqref{fanns} from \eqref{frukkklk3}, apply $\HACint$ to $(\forall^{\st}\alpha^{1}\leq_{1}1)(\exists^{\st}n)(\overline{\alpha}n\not\in T)$ to obtain standard $\Psi^{1\di 0^{*}}$ such that  
$(\forall^{\st}\alpha^{1}\leq_{1}1)(\exists n\in \Psi(\alpha))(\overline{\alpha}n\not\in T)$, and defining $g(\alpha):=\max_{i<|\Psi|}\Psi(\alpha)(i)$ we obtain $g$ as in the antecedent of \eqref{frukkklk3}.  The previous implies 
\be\label{gundark}
(\forall T^{1}\leq_{1}1) \big[(\forall^{\st}\alpha^{1}\leq_{1}1)(\exists^{\st}n)(\overline{\alpha}n\not\in T)\di (\exists^{\st}k)(\forall \beta\leq_{1}1)(\overline{\beta}i\not\in T) \big], 
\ee
which is the contraposition of \eqref{fanns}, using classical logic.  For the implication $\eqref{fanns}  \di \eqref{frukkklk}$, consider the contraposition of \eqref{fanns}, i.e.\ \eqref{gundark}, and note that the latter implies \eqref{frukkklk3}.  Now push all standard quantifiers outside as follows:
\[
(\forall^{\st}g^{2})(\forall T^{1}\leq_{1}1)(\exists^{\st} ( \alpha^{1}\leq_{1}1, ~k^{0})\big[(\overline{\alpha}g(\alpha)\not\in T)
\di(\forall \beta\leq_{1}1)(\exists i\leq k)(\overline{\beta}i\not\in T) \big], 
\]
and applying idealisation \textsf{I} yields \eqref{frukkklk}.  The final part now follows immediately in light of the basic axioms of $\P$ in Definition \ref{debs}.  
\end{proof}
By the previous theorem $\Theta$ emerges from Nonstandard Analysis, and the behaviour of $\Theta$ as in Theorem \ref{import} can be explained similarly: 
It is part of the folklore of Nonstandard Analysis that \emph{Transfer} does not imply \emph{Standard Part}.   The same apparently holds for fragments: $\P+\paai$ does not prove $\STP$ by the results in \cite{dagsam}*{\S6}.      
As a result of applying Theorem \ref{consresult}, there is no term of G\"odel's $T$ which computes $\Theta$ in terms of $(\mu^{2})$.  A stronger result as in Theorem \ref{import} apparently can be obtained.    

\medskip

Next, we discuss a nonstandard version of $\WWKL$,  introduced in \cite{pimpson}, as follows:
\be\tag{$\LMP$}
(\forall   T \leq_{1}1)\big[  \nu(T)\gg 0 \di (\exists^{\st} \beta\leq_{1}1)(\forall^{\st}n)(\overline{\beta}n\in T)\big], 
\ee
and one obtains a normal form of $\LMP$ similar to \eqref{frukkklk}.  As for the latter, this normal form gives rise to a \emph{weak fan functional} $\Lambda$, first introduced in \cite{dagsam}*{\S3}.  
We have the following theorem where $\ATR_{0}$ is the fourth `Big Five' system of RM (See \cite{simpson2}*{V}).  
\begin{thm}
The system $\P+\paai+\STP$ proves the consistency of $\ATR_{0}$ while $\P+\paai+\LMP$ does not.  
\end{thm}
\begin{proof}
See \cite{dagsam}*{\S6}.
\end{proof}
The previous result is referred to as a `phase transition' in \cite{dagsam} as there is (currently) nothing in between $\WKL_{0}$ and $\WWKL_{0}$ in the RM zoo. 
\begin{cor}
The system $\P+\paai$ does not prove $\STP$.  
\end{cor}
\begin{proof}
Suppose $\P+\paai$ proves $\STP$ and note that 
$\P+\paai$ then proves the consistency of $\ATR_{0}$ by the theorem.  By Theorem \ref{consresult}, $\textsf{E-PA}^{\omega*}+(\mu^{2})$ also proves the consistency of $\ATR_{0}$, which contradicts G\"odel's incompleteness theorems.
\end{proof}
At the end of Section \ref{FUWKL}, we noted that Theorem \ref{proto} should go through for $\WKL$ instead of $\WWKL$. 
In particular, the restriction to trees of positive measure (which is part of $\WWKL$) can be lifted while still obtaining the same equivalence as in \eqref{frood8}.   
Hence, these results are \emph{robust}, i.e.\ equivalent to small perturbations of themselves (See \cite{montahue}*{p.\ 432}).  We now provide an example where the notion of `tree with positive measure' yields \emph{non-robust} results.  
To this end, consider the following strengthening of $\LMP$:
\be\tag{$\LMP^{+}$}
(\forall T \leq_{1}1)\big[ \mu(T)>_{\R}0\di (\exists^{\st} \beta\leq_{1}1)(\forall^{\st} m)(\overline{\beta}m\in T) \big],
\ee
and the following weakening of $\STP$:
\be\tag{$\STP^{-}$}
(\forall T \leq_{1}1)\big[(\forall n^{0})(\exists \beta^{0^{*}})(\beta\in T\wedge |\beta|=n)\di (\exists^{\st} \beta\leq_{1}1)(\forall^{\st} m)(\overline{\beta}m\in T) \big],
\ee
%We have the following theorem.  
\begin{thm}\label{komo}
In $\P_{0}+\WWKL$, we have $\STP\asa \LMP^{+}\asa \STP^{-}$.
\end{thm}
\begin{proof}
The implication $\STP\di \STP^{-}$ is immediate; for the reverse implication, apply overspill to the antecedent of \eqref{fanns} to obtain a sequence $\beta_{0}^{0^{*}}$ of nonstandard length in $T$.  Extend the latter to an infinite tree by including $\beta_{0}*00\dots$, which is nonstandard.   
Applying $\STP^{-}$ to this extended tree yields the consequence of \eqref{fanns}, and hence $\STP$.
For $\STP\di \LMP^{+}$, apply $\STP$ to the path claimed to exist by $\WWKL$ and note that we obtain $\LMP^{+}$.  
For $\LMP^{+}\di\STP$, fix $f^{1}\leq_{1}1$ and nonstandard $N$.  Define the tree $T\leq_{1}1$ which is $f$ until height $N$, followed by the full binary tree.  
Then $\mu(T)>_{\R}0$ and let standard $g^{1}\leq_{1}1$ be such that $(\forall^{\st}n)(\overline{g}n\in T)$.   By definition, we have $f\approx_{1}g$, and we are done.
\end{proof}
In light of the theorem, $\STP$ and $\Theta$ seem fairly robust, while $\LMP$ and $\Lambda$ are not.

\medskip

Finally, $\STP$ and $\LMP$ are not unique: Similar nonstandard (and functional) versions exist for most of the theorems in the RM zoo.  Indeed, since every theorem $T$ in the RM zoo follows from arithmetical comprehension, 
we can prove $T^{\st}$ in $\RCAO+\paai+\STP$, and use $\STP$ to drop the `st' in the leading quantifier as in \eqref{fanns} for $\WKL^{\st}$.  
The author and Dag Normann are currently investigating the exact power of $\Theta$ and $\Lambda$ and the strength of the associated nonstandard axioms.  
%\begin{ack}\rm
%This research was supported by FWO Flanders, the John Templeton Foundation, the Alexander von Humboldt Foundation, the University of Oslo, and the Japan Society for the Promotion of Science.  
%The author expresses his gratitude towards these institutions. 
%The author thank Ulrich Kohlenbach for his valuable advice.   The author also thanks the referees for their valuable comments which greatly improved this paper.  
%\end{ack}

%\section{Bibliography}

\bye